\documentclass[11pt]{article}

\usepackage[utf8]{inputenc} 
\usepackage[T1]{fontenc}    
\usepackage{hyperref}       
\usepackage{url}            
\usepackage{booktabs}       
\usepackage{amsfonts}       
\usepackage{nicefrac}       
\usepackage{microtype}      
\usepackage{xcolor}         
\usepackage{fullpage}
\usepackage[numbers]{natbib} 

\usepackage{microtype}
\usepackage{booktabs} 

\usepackage{hyperref}

\usepackage{amsmath,amssymb,amsfonts,amsthm,bm}
\usepackage[capitalize,noabbrev]{cleveref}
\usepackage{xspace}
\usepackage[shortlabels]{enumitem}
\usepackage[showonlyrefs]{mathtools}
\usepackage{wrapfig}
\usepackage{caption}
\usepackage{graphicx}

\usepackage{multicol,multirow,array, bbm,varwidth,semantic}
\usepackage{tikz}

\usepackage[utf8]{inputenc}
\usepackage[utf8]{inputenc}
\usepackage[T1]{fontenc}
\usepackage[english]{babel}
\usepackage{setspace}
\usepackage{verbatim}
\usepackage{booktabs,siunitx}
\usepackage{subfig}
\usepackage{xcolor}
\usepackage{enumitem}
\usepackage{subfloat} 

\usepackage{algpseudocode}

\hypersetup{
    colorlinks=true,
    citecolor=cyan,
    linkcolor=blue,
    filecolor=magenta,      
    urlcolor=cyan,
}
\usepackage{xcolor}

\newtheorem{theorem}{Theorem}
\newtheorem{corollary}[theorem]{Corollary}
\theoremstyle{definition}
\newtheorem{definition}{Definition}
\newtheorem{lemma}[theorem]{Lemma}

\DeclareMathOperator{\R}{\mathbb{R}}

\newcommand{\E}{\mathbb{E}}
\newcommand{\cL}{\mathcal{L}}
\newcommand{\cW}{\mathcal{W}}

\newcommand{\cO}{\mathcal{O}}

\newcommand{\cX}{\mathcal{X}}
\newcommand{\bx}{\mathbf{x}}
\newcommand{\eps}{\varepsilon}
\newcommand{\VIEW}{\textsc{view}}
\newcommand{\Z}{\mathbb{Z}}
\newcommand{\NB}{\mathrm{NB}}
\newcommand{\sNB}{\mathrm{sNB}}
\newcommand{\OPT}{\mathrm{OPT}}
\newcommand{\LP}{\mathrm{LP}}
\renewcommand{\setminus}{\smallsetminus}
\newcommand{\cP}{\mathcal{P}}

\newcommand{\supp}{\mathrm{supp}}
\newcommand{\bt}{\mathbf{t}}
\newcommand{\bzero}{\mathbf{0}}
\newcommand{\optlp}{{\OPT}_{\LP}}
\newcommand{\optp}{{\OPT}_{\mathrm{pack}}}
\newcommand{\bdeg}{\boldmath{\deg}}
\newcommand{\DLap}{\mathrm{DLap}}

\newcommand{\bN}{\mathbb{N}}
\newcommand{\ind}{\mathbf{1}}
\newcommand{\ta}{\tilde{a}}
\newcommand{\ba}{\mathbf{a}}
\newcommand{\by}{\mathbf{y}}
\newcommand{\optind}{\rho(G)}
\newcommand{\tbx}{\tilde{\bx}}
\newcommand{\tx}{\tilde{x}}
\newcommand{\tP}{\tilde{P}}
\newcommand{\packnum}{\rho}

\usepackage{array}
\usepackage{multirow}
\usepackage{tabularray}

\allowdisplaybreaks

\title{Differential Privacy on Trust Graphs}

\author{Badih Ghazi \\
Google Research \\
\and 
Ravi Kumar \\
Google Research \\
\and
Pasin Manurangsi \\
Google Research \\
\and 
Serena Wang \\
Google Research \\
}

\date{}

\begin{document}

\maketitle

\begin{abstract}
We study differential privacy (DP) in a multi-party setting where each party only trusts a (known) subset of the other parties with its data. Specifically, given a \emph{trust graph} where vertices correspond to parties and neighbors are mutually trusting, we give a DP algorithm for aggregation with a much better privacy-utility trade-off than in the well-studied local model of DP (where each party trusts no other party). We further study a robust variant where each party trusts all but an unknown subset of at most $t$ of its neighbors (where $t$ is a given parameter), and give an algorithm for this setting. We complement our algorithms with lower bounds, and discuss implications of our work to other tasks in private learning and analytics.
\end{abstract}

\section{Introduction}

\emph{Differential privacy} (DP) \citep{dwork2006calibrating, dwork2006our} is a rigorous privacy notion that has seen extensive study (e.g., \citep{dwork2014algorithmic, vadhan2017complexity}) and widespread adoption in analytics and learning (e.g., \citep{ding2017collecting, abowd2018us, tf-privacy, pytorch-privacy}). It dictates that the output of a randomized algorithm remains statistically indistinguishable if the data of a single user changes. The most widely studied models of DP are the \emph{central} model where a trusted curator is given access to the raw data and required to output a DP estimate of the function of interest, and the \emph{local} model \citep{evfimievski2003limiting, kasiviswanathan2011can} where every message leaving each user’s device is  required to be DP. While the latter is 
 compelling in that each user needs to place minimal trust in other users, it is known to suffer from a significantly higher utility degradation compared to the former (e.g., \citep{ChanSS12}).

In practice, data sharing settings often include situations where a user is willing to place more trust in a \emph{subset} of other users. For example, many people have different privacy sensitivities depending on their relationships with others: 
Alice might be willing to share her location data with family and close friends, but unwilling to have her location data be recoverable by strangers from a public channel. This relates to philosophical conceptualizations of privacy as \textit{control over personal information}, in that individuals may specify with whom they are willing share their information \citep{solove2002conceptualizing} (Appendix~\ref{app:motivating_examples}).

\paragraph{Trust Graph DP Models.}
In this work, we model such relationships as a \emph{trust graph} where vertices correspond to different users and neighboring vertices are mutually trusting (see Figure \ref{fig:network} for a simple example). 
We define and study DP over such trust graphs, whereby the DP guarantee is enforced on the messages exchanged by each vertex or its trusted neighbors on one hand, and the other non-trusted vertices on the other hand. 



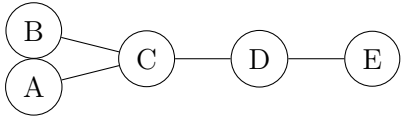
\begin{figure}[!ht]
    \centering
\begin{minipage}[c]{0.4\textwidth}
    \begin{tikzpicture}[scale=0.75]
    \node[shape=circle,draw=black] (A) at (0,0.5) {A};
    \node[shape=circle,draw=black] (B) at (0,1.5) {B};
    \node[shape=circle,draw=black] (C) at (2,1) {C};
    \node[shape=circle,draw=black] (D) at (4,1) {D};
    \node[shape=circle,draw=black] (E) at (6,1) {E};

    \path [-] (A) edge node[left] {} (B);
    \path [-](B) edge node[left] {} (C);
    \path [-](C) edge node[left] {} (D);
    \path [-](C) edge node[left] {} (A);
    \path [-](D) edge node[left] {} (E);
\end{tikzpicture}
  \end{minipage}\hfill
  \begin{minipage}[c]{0.6\textwidth}
    \caption{Simple example trust graph. User A is only willing to share their data with users B and C, and user C is additionally willing to share their data with D. We introduce a privacy model (TGDP) in which users D and E cannot identify user A's data based on any communication exchanged. }
    \label{fig:network}
  \end{minipage}
\end{figure}

Specifically, we define notions of \textit{Trust Graph DP (TGDP)} that  generalize existing definitions of local DP \citep{kasiviswanathan2011can} and central DP, and effectively interpolate between them. Informally, TGDP requires that the distribution of all messages exchanged by each vertex $v$ or one of its neighbors with any vertices that are not trusted by $v$ should remain statistically indistinguishable if the input data held by $v$ changes; we formalize this in Definition~\ref{def:graph-dp}.

We further extend TGDP to capture \textit{robustness} to potentially compromised neighbors. Namely, the above privacy guarantee can break if a single neighbor of a vertex turns out to be untrustworthy. Thus, we introduce the notion of \textit{Robust Trust Graph DP (RTGDP)}, which maintains the privacy guarantee even if some unknown subset of neighbors of a given size is compromised; see Definition~\ref{def:robust-graph-dp}.

\paragraph{Our Results.} 
Having defined these trust graph-based models of DP, we give algorithms for the basic \emph{aggregation} primitive that satisfy both TGDP and RTGDP notions. Notably, we propose algorithms that depend on linear programming formulations that can be computed in polynomial time (Theorems~\ref{thm:lp-protocol} and~\ref{thm:lp-protocol-robust}). We complement our algorithms with lower bounds on the error that depend on combinatorial properties of the graph (Theorems~\ref{thm:tgdplb} and~\ref{thm:robust-lb-generic}). Although closing the gap between our upper and lower bounds is still open, we obtain a \emph{bi-criteria} result showing that the upper bound is not much larger than the lower bound when we slightly increase the robustness parameter $t$ (Theorem~\ref{thm:robust-bicriteria-gap}).

The aggregation primitive we study in this work is a basic building block. Indeed, our work implies new DP algorithms over trust graphs for other problems in  learning and  analytics (see Appendix~\ref{app:ML}).

We supplement the theory with evaluation on nine real network datasets including email communication networks, social networks, and cryptocurrency trust networks. Our results show that the utility degradation when satisfying TGDP and RTGDP can be significantly lower than that of local DP. Thus, when trust relationships exist, accounting for these relationships when incorporating DP into a system can considerably improve overall utility.


\subsection{Technical Overview}

We first give a TGDP mechanism for the integer aggregation problem with a mean-squared error (MSE) that scales linearly in the size of any dominating set\footnote{See \Cref{sec:trust_graph_DP} for the formal definitions of a dominating set and a packing.} of the trust graph. The idea of the protocol is the following: Each user identifies one user in the dominating set that they trust and to whom they send their input. Then, each user in the dominating set simply runs the central (discrete) Laplace mechanism to privatize the data. The final estimate is the sum of all these sub-estimates in the dominating set. The MSE grows with the dominating set size, as formalized in Theorem~\ref{thm:domset-protocol}.

A disadvantage of the protocol described above is that, in order to minimize the error, it requires the knowledge of a minimum dominating set. However, computing a minimum dominating set is  NP-hard and even hard to approximate~\citep{Feige98}. So we give (in Theorem~\ref{thm:lp-protocol}) a protocol that not only is efficient but can also reduce the error by up to $O(\log n)$ factor. This protocol only requires a solution to a linear program (LP), which, unlike the minimum dominating set, can be computed in polynomial time. At a high-level, there are two key ideas in this protocol: 

\medskip
\noindent \emph{(i) Input Splitting:} Instead of sending the input to a single vertex as in the previous section, each user will split their input into (random) additive shares and send it to all its neighbors. The input splitting idea originated in cryptography~\citep{IKOS06} and has recently found applications in the shuffle model of DP~\citep{BBGN20,GMPV20}, although the nature of how we use it here is quite different from those previous works.

\medskip
\noindent\emph{(ii) Distributed Noise Addition:} Similar to the previous protocol, each user again broadcasts the sum of all messages they receive with some noise added. The main difference here is that, instead of using the discrete Laplace noise, we use the negative binomial noise designed in such a way that, when sufficiently many of them are summed up, they guarantee DP. This helps reduce the amount of noise required in the protocol. (The idea of distributed noise generation dates back to the early works on DP~\citep{dwork2006our}, but the distributions we use here are from~\citep{BBGN20}.)

\medskip

What we gain by applying the input splitting is that, due to the properties of random additive shares, the only way the adversary learns anything about $x_v$ is to sum up all the messages broadcast from its neighbors. By a careful design of the distributed noise distribution, we can ensure that this sum contains sufficient noise to provide DP guarantees.

Our lower bound on integer aggregation with TGDP (Theorem~\ref{thm:tgdplb}) shows that the MSE grows with the packing number of the trust graph.  The main idea is to  transform any TGDP protocol to a local DP (LDP) protocol with the same privacy
and utility, but with a number of users equal to the packing number.  The ``packing'' property ensures that the users are ``isolated'' from each other in the reduction step.

For our results in the RTGDP model, we consider the same LP but impose a stricter constraint to ensure DP guarantees even when some neighbors of each vertex are compromised. To prove our bi-criteria tightness (Theorem~\ref{thm:robust-bicriteria-gap}), we study the dual of the LP and apply randomized rounding to convert the fractional solution into an integral one.

\subsection{Related Work}\label{ref:related_work}

Secure multiparty computation (SMPC) \citep{yao1982protocols,yao1986generate,goldreich2019play} can be leveraged to allow users to achieve central DP utility without relying on a trusted curator; this is done via cryptographic protocols whose security relies on computational hardness assumptions~\citep{dwork2006our,BeimelNO08,bonawitz2017practical,bell2020secure}. 
An important distinction between our model and SMPC is that while the privacy of SMPC protocols relies on computational hardness assumptions, the privacy guaranteed in our proposed TGDP model is information-theoretic (and thus stronger). Still, our proposed TGDP model can also be thought of as relaxing the SMPC threat model to assume that a subset of at least a certain number of users is  trustworthy (or would execute the algorithm faithfully). Specifically, the proposed TGDP notion could enable higher-utility protocols than the state-of-the-art in SMPC, by making a stronger (but realistic) assumption that different users fully trust some subset of other users (namely, their neighbors in the trust graph). 

We also point out that the multi-central (a.k.a. multi-server) model of DP \citep{steinke2020multi, cheu2023necessary} is a special case of our RTGDP model. Another intermediate model between local and central DP is the shuffle model \citep{bittau17,erlingsson2019amplification,CheuSUZZ19}, where aggregation has been extensively studied (e.g., \citep{balle2019privacy, ghazi2020private, ghazi2021differentially}); but this model does not capture mutually trusting relationships between different pairs of users. Finally, the network DP model \citep{cyffers2022privacy} is a different relaxation of local DP where the (DP) communication between users is restricted to the edges of a given graph; this is in contrast to our proposed TGDP model where a given graph encodes trust relationships.

\paragraph{Organization.}
We start with some background in Section~\ref{sec:prelim}. In Section~\ref{sec:trust_graph_DP}, we formally define the notion of DP on trust graphs, and give algorithms and a lower bound for solving the aggregation task under this privacy guarantee. The RTGDP notion is defined in Section~\ref{ref:robust_trust_graph_DP} where we also give an algorithm and lower bounds for aggregation under this robust notion. 
We conclude with some interesting future directions in Section~\ref{sec:conc_fut_dir}. In Appendix~\ref{app:motivating_examples}, we provide additional motivating examples for TGDP.  Missing proofs are  in Appendix~\ref{app:proofs}.  In Appendix~\ref{app:exp}, we include experiments in which we report our given upper and lower bounds on real network datasets.
\section{Preliminaries}\label{sec:prelim}

Let $n$ users be represented as vertices $V = \{1,...,n\}$ of a graph $G = (V,E)$, where  $E \subseteq V^2$ corresponds to the set of pairs $(i, j)$ of users, where $i$ and $j$ are willing to share their data with each other.%
\footnote{For simplicity, we focus only on the ``symmetric'' notion of sharing. It is relatively simple to extend all of our algorithms to the ``asymmetric'' version as well.}
Let $\mathcal{X}$ be any domain and suppose each user $i \in V$ has data $x_i \in \mathcal{X}$, and the (full) input dataset is given by $(x_1, \ldots, x_n) = \mathbf{x} \in \mathcal{X}^n$. Let $N(v)$ be the \emph{neighborhood} of $v$ in $G = (V, E)$, i.e., $N(v) = \{ u \mid (u, v) \in E \}$ and let $N[v]$ be the \emph{closed neighborhood} of $v$, i.e., $N[v] = N(v) \cup \{v\}$.

\subsection{Differential Privacy Definitions and Tools}


\begin{definition}\label{def:dp} (DP;~\citep{dwork2006calibrating, dwork2006our})
A randomized mechanism $M: \mathcal{X}^n \to \cO$ is \emph{($\eps$, $\delta$)-differentially private ($(\epsilon, \delta)$-DP)} if for all pairs $\mathbf{x}, \mathbf{x}' \in \mathcal{X}^n$ of datasets that differ only in the data of a single user, and for all subsets $S \subseteq \cO$, $\Pr[M(\mathbf{x}) \in S] \leq e^\eps \Pr[M(\mathbf{x}') \in S] + \delta$.
\end{definition}

For brevity, we write $(\eps, 0)$-DP as $\eps$-DP (a.k.a., \emph{pure}-DP).

In \emph{non-interactive local DP}, each user has to randomize their own input and send it to the server (or, alternatively, publish it). In this case, each user's randomized output is required to be DP: 

\begin{definition}\label{def:nonint-ldp} (Non-Interactive Local DP;~\citep{kasiviswanathan2011can})
A randomized mechanism $M: \mathcal{X} \to \cO$ is a \emph{non-interactive ($\eps$, $\delta$)-local DP randomizer} if for any pair $x, x' \in \mathcal{X}$, and for all subsets $S \subseteq \cO$, $\Pr[M(x) \in S] \leq e^\eps  \Pr[M(x') \in S] + \delta$.
\end{definition}

To define DP properties of possibly interactive protocols, we follow the approach of \citep{BeimelNO08}. First, we define the notion of a protocol view. 

\begin{definition}[View]
The \emph{view} of a protocol $P$ at vertex $u$ for an input dataset $\bx \in \mathcal{X}^n$, denoted $\VIEW^{u}_P(\bx)$, consists of the input $x_u$ and all  messages received and sent (together with the corresponding source/destination) by the vertex $u$.
\end{definition}
The view of the protocol $P$ for a subset $S \subseteq V$ of vertices is defined as $\VIEW^S_P(\bx) := (\VIEW^u_P(\bx))_{u \in S}$. Let $\cO$ be the set of all possible views. For any $T \subseteq V$, we write $\bx_{-T}$ as a shorthand for $\bx_{V \setminus T}$ and, for any $v \in V$, $\bx_{-v}$ as a shorthand for $\bx_{-\{v\}}$. When the protocol is interactive, local DP can be defined as follows~\cite{BeimelNO08}:
\begin{definition} \label{def:int-ldp} (Interactive Local DP)
A protocol $P$ satisfies \emph{$(\epsilon,\delta)$-local DP ($(\epsilon, \delta)$-LDP)} if for each vertex $v \in V$, $\VIEW^{V \setminus \{v\}}_P(\bx)$ satisfies $(\eps,\delta)$-DP with respect to the input $x_v$ for all values of $\bx_{-v}$. I.e., for all pairs $x_v, x'_v \in \mathcal{X}$, all values of $\bx_{-v}$, and all subsets $S \subseteq \cO$,
\[\Pr[\VIEW^{V \setminus \{v\}}_P(x_v, \bx_{-v}) \in S]\leq e^\eps \Pr[\VIEW^{V \setminus \{v\}}_P(x'_v, \bx_{-v}) \in S] + \delta.\]
\end{definition}

For pure-DP, it is useful to define $D_{\infty}(\cP ~\|~ \cP') := \max_{o \in \supp(\cP)} \ln\left(\frac{\Pr_{X \sim \cP}[X = o]}{\Pr_{X' \sim \cP}[X' = o]}\right)$ for distributions $\cP, \cP'$. We will sometimes use random variables and distributions interchangeably. 

It is well-known that DP is robust to post-processing. 
This fact will be useful in our privacy analysis.

\begin{lemma}[Post-Processing] \label{lem:postp}
For any random variables $X, X'$ and a (possibly randomized) function $f$, we have $D_{\infty}(f(X) ~\|~ f(X')) \leq D_{\infty}(X ~\|~ X')$.
\end{lemma}

We will use the following  distributions for the noise:
\begin{itemize}
\item The \emph{negative binomial distribution} $\NB(r, p)$ with parameters $r > 0, p \in (0, 1)$ is supported on $\mathbb{Z}_{\geq 0}$ with density $\Pr[X = k] = {k+r-1 \choose k} (1-p)^k p^r,$ where $X \sim \NB(r, p)$.  Its variance is $r(1-p)/p^2$. 
\item Let $\sNB(r, p)$ be the distribution of $X - X'$ where $X, X' \sim \NB(r, p)$ are i.i.d.
\item The \emph{discrete Laplace distribution} $\DLap(b)$ with parameter $b > 0$ is supported on $\mathbb{Z}$ and its density is given by $\Pr[X = k] \propto \exp(-|k|/b)$ for $X \sim \DLap(b)$.
\end{itemize}
We will use the following facts in our analysis:
\begin{itemize}
\item If $X_1 \sim \sNB(r_1, p)$ and $X_2 \sim \sNB(r_2, p)$, then $X_1 + X_2 \sim \sNB(r_1 + r_2, p)$.
\item $\DLap(b)$ is the same distribution as $\sNB(1, 1 - e^{-1/b})$.
\end{itemize}

The discrete Laplace mechanism is well-known to guarantee DP in the central setting~\cite{GhoshRS12}. Below, we state a slightly more general version of this for $\sNB$ that will be convenient for our analysis. 

\begin{lemma} \label{lem:snb-dp}
For any $x, x' \in \{0, \dots, \Delta\}$, let $Z \sim \sNB(r, 1 - e^{-\eps/\Delta})$ where $r \geq 1$. Then, we have
$D_{\infty}(Z + x ~\|~ Z + x') \leq \eps.$
\end{lemma}

In the privacy analysis, we often consider $\VIEW^S_P(x_v, \bx_{-v})$ and $\VIEW^S_P(x'_v, \bx_{-v})$ for $x_v, x'_v \in \cX$. For  convenience, we will write $\VIEW^S_P(x)$ for $x \in \cX$ as a shorthand for $\VIEW^S_P(x_v, \bx_{-v})$ when $x_v = x$. Similarly, for a quantity $y$ that depends on $x_v$, we will write $y(x)$ to denote $y$ when $x = x_v$.




\section{Trust Graph Differential Privacy}\label{sec:trust_graph_DP}

We model trust relationships across users as a network where vertices correspond to users, and undirected edges connect users who are  mutually willing to share their data (see Figure \ref{fig:network}). We focus on undirected trust graphs, though extensions of our results to directed graphs are possible.
For a given trust graph, we define a general notion of Trust Graph DP, provide algorithms for achieving it, and analyze upper and lower bounds on the error for the integer aggregation problem.


\begin{definition} (Trust Graph DP) \label{def:graph-dp}
Let $G = (V, E)$.  A protocol $P$ satisfies \emph{$(\epsilon,\delta, G)$-Trust Graph DP ($(\epsilon, \delta, G)$-TGDP)} if for each vertex $v \in V$, $\VIEW^{V \setminus N[v]}_P(\bx)$ satisfies $(\eps,\delta)$-DP with respect to the input $x_v$ for all values of $\bx_{-v}$. I.e., for all pairs $x_v, x'_v \in \mathcal{X}$, all values of $\bx_{-v}$, and all subsets $S \subseteq \cO$,
\[\Pr[\VIEW^{V \setminus N[v]}_P(x_v, \bx_{-v}) \in S]\leq e^\eps \Pr[\VIEW^{V \setminus N[v]}_P(x'_v, \bx_{-v}) \in S] + \delta.\]
\end{definition}
Referring back to Figure \ref{fig:network} as an example, Definition \ref{def:graph-dp} says that even if users D and E pooled their messages together, their collective view would still be DP with respect to the data for user A.

Notably, the proposed TGDP model generalizes both the central DP model and the local DP model. The central DP model is captured when $G$ is a star graph in which all vertices (the users) entrust their data a single central vertex (the analyst): each user's data is private relative to the view of all other users, but all users trust the same central analyst (see Figure \ref{fig:network_star} in the Appendix). The local DP model, on the other hand, is captured when $G$ simply has no edges between any vertices. Our TGDP model thus introduces a flexibility to capture intermediate trust relationships, perhaps involving several local analysts, or more general trust graphs arising from social networks.

As before, we write $(\eps, 0, G)$-TGDP as $(\eps, G)$-TGDP. Note that the $(\eps, G)$-DP condition in \Cref{def:graph-dp} can be written as $D_{\infty}\left(\VIEW^{V \setminus N[v]}_P(x_v) ~\middle\|~ \VIEW^{V \setminus N[v]}_P(x'_v)\right) \leq \eps$.

Recall that for a graph $G = (V, E)$, a  \emph{dominating set} is a subset $U \subseteq V$ such that for every $v \in V \setminus U$, there is a $u \in U$ such that $(u, v) \in E$; the size of a minimum dominating set is the \emph{domination number}  $\gamma(G)$. A \emph{packing} of $G$ is a subset $U \subseteq V$ such that for any distinct $u, u' \in U$, $N[u]$ and $N[u']$ are disjoint; the size of a maximum packing is the \emph{packing number} $\rho(G)$. 

\paragraph{Aggregation.}
We consider the \emph{integer aggregation} problem. 
Let each individual have a value $x_i \in \{0, \dots, \Delta\}$. The goal is to compute an estimate $\ta$ of $a = \sum_{i=1}^n x_i$. We measure the \emph{mean-square error} (MSE), which is defined as $\E[(\ta - a)^2]$,  where the expectation is over the randomness of the protocol.
In central DP, the standard Laplace mechanism~\cite{dwork2006calibrating} achieves an error of $2\Delta^2/\eps^2$. In local DP, the local version of the Laplace mechanism achieves an error of $2\Delta^2 n / \eps^2$. Both of these are known to be asymptotically optimal.

\subsection{Algorithm via Dominating Set}\label{subsec:dom}

We start by giving a protocol for the integer aggregation problem using the graph's dominating set.

\begin{theorem} \label{thm:domset-protocol}
There is an $(\eps, G)$-TGDP mechanism for the aggregation problem with MSE at most $2\Delta^2 |T| / \eps^2$, where $T$ is any dominating set of $G$.
\end{theorem}

\begin{proof}
The protocol works as follows:
\begin{itemize}
\item First, each user $v \in V$ picks an arbitrary vertex $u_v \in T \cap N[v]$. (The intersection is not empty since $T$ is a dominating set.) Then, the user sends $x_v$ to $u_v$.
\item Each user $u \in T$ broadcasts the sum of all numbers  it receives together with a noise drawn from $\DLap(\Delta/\eps)$. 
More formally, the user broadcasts $a_u = \sum_{v \in V \atop u_v = u} x_v + z_u$, where $z_u \sim \DLap(\Delta/\eps)$.
\item Finally, the estimate is $\ta = \sum_{u \in T} a_u$.
\end{itemize}

\paragraph{Privacy Analysis.}
Consider any $v \in V$ and $\bx_{-v} \in \{0, \dots, \Delta\}^{V \setminus \{v\}}$. We write $\ba$ as a shorthand for $(a_u)_{u \in T}$. 
Let $S_v := \{w \in N[v] \mid u_w \in N[v]\}$ denote the nodes in $N[v]$ whose message in the first step is sent to a node in $N[v]$.
Notice that $\VIEW^{V \setminus N[v]}_P(x)$ is exactly $(\bx_{-S_v}, \ba(x))$.
We claim that this is a post-processing of $z_{u_v} + x$. This is simply because $\bx_{-S_v}, (a_u)_{u \in T \setminus \{u_v\}}$ do not depend on $x_v = x$ at all and are independent of $z_{u_v} + x$; finally, note that $a_{u_v}(x)$ is a post-processing of $z_{u_v} + x$ since $a_{u_v}(x) = \left(z_{u_v} + x\right) + \sum_{v' \in V \setminus \{v\} \atop u_{v'} = u_v} x_{v'}$. 

Consider any $x_v, x'_v \in \{0, \dots, \Delta\}$. By \Cref{lem:postp} and \Cref{lem:snb-dp}, we have
\begin{equation*}
D_{\infty}\left(\VIEW^{V \setminus N[v]}_P(x_v) ~\middle\|~ \VIEW^{V \setminus N[v]}_P(x'_v)\right) \leq D_{\infty}(z_{u_v} + x_v ~\|~ z_{u_v} + x'_v) \leq \eps,
\end{equation*}
where the second inequality is due to \Cref{lem:snb-dp}.  Thus, the protocol satisfies $(\eps, G)$-TGDP as desired.

\paragraph*{Utility Analysis.}
The MSE is $\E\left[\left(\ta - a\right)^2\right] 
= \sum_{u \in T} \E[z_u^2]
\leq |T| \cdot \frac{2\Delta^2}{\eps^2}.$
\end{proof}

\subsection{Improved Algorithm via Linear Programming}

A disadvantage of the protocol from Section~\ref{subsec:dom} is that to minimize the error, it requires the knowledge of a minimum dominating set. Computing minimum dominating set is  NP-hard and even hard to approximate~\cite{Feige98}. In this section, we give a protocol that is efficient to compute and furthermore can reduce the error by up to $O(\log n)$ factor in certain graphs.  To describe our protocol, recall the linear programming (LP) relaxation of the dominating set problem:
\begin{equation}\label{eq:origlp}
\min \;\;\sum_{u \in V} y_u \qquad \text{s.t.}\;\; \sum_{u \in N[v]} y_u \geq 1 \quad \forall v \in V; \qquad 0 \leq y_u \leq 1 \quad \forall u \in V.
\end{equation}
To see that this is a relaxation of the dominating set problem, note that any dominating set $T \subseteq V$ gives a solution by setting $y_v = \ind[v \in T]$. Due to this, the optimum of this LP is no more than the size of the dominating set. In fact, the LP optimum can be smaller than the minimum dominating set size by an $O(\log n)$ factor~\cite{Lovasz75}.

The main result is a protocol whose MSE scales with the LP optimum instead of dominating set:

\begin{theorem}\label{thm:lp-protocol}
There is an $(\eps, G)$-TGDP mechanism for the aggregation problem with MSE at most $2\Delta^2\cdot\OPT_{\LP}/\eps^2$, where  $\OPT_{\LP}$ denotes the value of the optimal solution to the LP in \eqref{eq:origlp}.
\end{theorem}

\begin{proof}
Let $\by = (y_u)_{u \in V}$ denote any solution to the LP in \eqref{eq:origlp}.
The protocol works as follows:
\begin{itemize}
\item Let $q = 2n\Delta$.
\item For every user $v \in V$, pick $\{s^u_v\}_{u \in N[v]} \subseteq \Z_q$ uniformly at random among those that satisfy $\sum_{u \in N[v]} s^u_v \equiv x_v \mod q$. Then, for every $u \in N[v]$, user $v$ sends $s^u_v$ to $u$.
\item For every $u \in V$, sample $z_u \sim \sNB(y_u, 1 - e^{-\eps/\Delta})$; broadcast $a_u \equiv z_u + \sum_{v \in N[u]} s^u_v \mod q$.
\item Compute $a' \equiv \sum_u a_u \mod q$. Then, output
$
\ta =
\begin{cases}
a' & \text{ if } a' \leq q/2, \\
a' - q &\text{ otherwise.}
\end{cases}
$
\end{itemize}

\paragraph*{Privacy Analysis.}
Throughout the analysis, we assume that the addition is modulo $q$ unless stated otherwise.
Consider any $v \in V$ and $\bx_{-v} \in \{0, \dots, \Delta\}^{V \setminus \{v\}}$. We write $\ba$ as a shorthand for $(a_u)_{u \in V}$.
Notice that $\VIEW^{V \setminus N[v]}_P(x)$ is exactly $(\bx_{-N[v]}, \ba(x), (s^u_{v'})_{u \in V \setminus N[v], v' \in V})$.
We claim that this is a post-processing of $(z_u + s^u_v)_{u \in N[v]}$. This is simply because $\bx_{-N[v]}, (a_u)_{u \in V \setminus N[v]}, (s^u_{v'})_{u \in V \setminus N[v], v' \in V}$ do not depend on $x_v = x$ at all and are independent of $(z_u + s^u_v)_{u \in N[v]}$; finally, note that $(a_u(x))_{u \in N[v]}$ is a post-processing of $(z_u + s^u_v)_{u \in N[v]}$ since $a_{u}(x) = \left(z_u + s^u_v\right) + \sum_{v' \in N[v] \setminus \{v\} \atop u_{v'} = u} s^u_{v'}$ for all $u \in N[v]$.

For any $x_v, x'_v \in \{0, \dots, \Delta\}$, \Cref{lem:postp} implies that
%
\begin{equation*}
D_{\infty}\left(\VIEW^{V \setminus N[v]}_P(x_v) ~\middle\|~ \VIEW^{V \setminus N[v]}_P(x'_v)\right) \leq D_{\infty}((z_u + s^u_v(x_v))_{u \in N[v]} ~\|~ (z_u + s^u_v(x'_v))_{u \in N[v]}).
\end{equation*}
Now, since $(s^u_v(x))_{u \in N[v]}$ are random elements of $\Z_q$ that sum to $x$, we also have that $(z_u + s^u_v(x))_{u \in N[v]}$ are random elements of $\Z_q$ that sum to $x + \sum_{u \in N[v]} z_u$. In other words, $(z_u + s^u_v(x_v))_{u \in N[v]}$ is a post-processing of $x + \sum_{u \in N[v]} z_u$. Again,  \Cref{lem:postp} implies that
%
\begin{equation*}
D_{\infty}((z_u + s^u_v(x_v))_{u \in N[v]}~\|~(z_u + s^u_v(x'_v))_{u \in N[v]}) \leq D_{\infty}\left(x_v + \sum_{u \in N[v]} z_u~\middle\|~x'_v + \sum_{u \in N[v]} z_u\right).
\end{equation*}
Finally, $Z := \sum_{u \in N[v]} z_u$ is distributed as $\sNB\left(\sum_{u \in N[v]} y_u, 1 - e^{-\eps/\Delta}\right)$. 
Since $\by$ is feasible in \eqref{eq:origlp},
we have $\sum_{u \in N[v]} y_u \geq 1$. Thus, we can apply \Cref{lem:snb-dp} to conclude that the RHS above is $\le \eps$.

\paragraph*{Utility Analysis.} 
Note $\ta \equiv a + \left(\sum_{u \in V} z_u\right) \mod q$.
Since $a \in [0, q/2]$ and $\ta \in (-q/2, q/2]$, we have $|a - \ta| \le \left|\sum_{u \in V} z_u\right|$. Thus, the MSE is $\leq 
\sum_{u \in T} \E[z_u^2] \leq \sum_{u \in T} \frac{2y_u}{(\eps/\Delta)^2} = \frac{2 \Delta^2 \OPT_\LP}{\eps^2}$,
where the last equality is from our assumption that $(y_u)_{u \in V}$ is an optimal solution to the LP in \eqref{eq:origlp}. 
\end{proof}

We remark that, in the proof above, the privacy guarantee holds even for non-optimal LP solution $\by$, as long as it satisfies the constraints. Similarly, the error guarantee holds where $\OPT_{\LP}$ is replaced with the objective value of the solution. This is helpful for practical applications where we may only have an approximately optimal LP solution.

\subsection{Lower Bound}

We now give a lower bound for integer aggregation, where the MSE grows with the packing number.

\begin{theorem}
\label{thm:tgdplb}
For any $\eps \leq O(1)$, any $(\eps, G)$-TGDP protocol for integer aggregation incurs MSE $\Omega(\Delta^2 \cdot \optind)$, where $\optind$ denotes the packing number of the trust graph $G$. 
\end{theorem}

In fact, we give the following reduction that transforms any TGDP protocol to an LDP protocol with the same privacy parameter and MSE, but only on $\optind$ users (instead of $n$ users). Applying the known $\Omega(\Delta^2 n)$ lower bound for integer aggregation in LDP~\cite{ChanSS12}\footnote{Note that~\citep{ChanSS12} state their lower bound for $\Delta = 1$ but the case $\Delta > 1$ follows by scaling up the input.} immediately yields \Cref{thm:tgdplb}.

\begin{lemma} \label{lem:red-lb}
Suppose that there is an $(\eps, G)$-TGDP protocol for integer aggregation. Then, there exists an $\eps$-local DP protocol for integer aggregation for $\optind$ users with the same MSE as the $(\eps, G)$-TGDP protocol, where $\optind$ denotes the packing number of $G$.
\end{lemma}

\begin{proof}
Let $U = \{u_1, \dots, u_m\} \subseteq V$ be the largest packing in $G$ where $m = \optind$. To avoid ambiguity, let $\tbx = (\tx_1, \dots, \tx_m)$ be the input to the LDP protocol (that we construct below).

To construct the LDP protocol, let $Q_1 \cup \cdots \cup Q_m$ be any partition of $V$ such that $N[u_i] \subseteq Q_i$ for all $i \in [m]$. Such a partition exists because $N[u_1], \dots, N[u_m]$ are disjoint by the definition of packing. Let $P$ be any $(\eps, G)$-TGDP protocol for integer aggregation. Our LDP protocol $\tP$ runs the protocol $P$ where each $\tP$'s user $i \in [m]$ assumes the role of all $P$'s users in $Q_i$, where the input to $P$ is defined as
$
x_u =
\begin{cases}
\tx_i &\text{ if } u = u_i \\
0 &\text{ otherwise,}
\end{cases}
\quad \forall u \in Q_i.$
We then output the estimate as produced by $P$. The MSE of $\tP$ is obviously the same as that of $P$.

To see that $\tP$ satisfies $\eps$-LDP, consider any $i \in [m], \tbx_{-i} \in \cX^{[m] \setminus \{i\}}$, we have $\VIEW^{[m] \setminus \{i\}}_{\tP}(\tx) = \VIEW^{V \setminus Q_i}_P(\bx(\tx))$, where $\bx(\tx)$ is the input to $P$ as defined above. Since $V \setminus Q_i \subseteq V \setminus N[u_i]$, $\VIEW^{V \setminus Q_i}_P(\bx(\tx))$ is a post-processing of $\VIEW^{V \setminus N[u_i]}_P(\bx(\tx))$, 
\Cref{lem:postp} implies that
%
\begin{equation*}
D_{\infty}\left(\VIEW^{[m] \setminus \{i\}}_{\tP}(\tx_i) ~\middle\|~ \VIEW^{[m] \setminus \{i\}}_{\tP}(\tx'_i)\right) \leq D_{\infty}\left(\VIEW^{V \setminus N[u_i]}_P(\tx_i) ~\middle\|~ \VIEW^{V \setminus N[u_i]}_P(\tx'_i)\right) 
\leq \eps,
\end{equation*}
where the last inequality is due to $P$ being an $(\eps, G)$-TGDP protocol.  Hence, $\tP$ is $\eps$-LDP.
\end{proof}
Unfortunately, the lower bound in~\Cref{thm:tgdplb} is not tight with respect to the upper bounds in~\Cref{thm:domset-protocol,thm:lp-protocol}.  Indeed, the following is a example of a graph that has a large gap between the domination number and the packing number~\cite{Burger}.  Let $V = [k] \times [k]$ for $k \in \mathbb{N}$. There is an edge between any $(x, y) \in V, (x', y') \in V$ iff $x = x'$ or $y = y'$. For this graph, $\OPT_{\LP}$ is
$\Omega(\sqrt{|V|})$  since every vertex has degree $O(k) = O(\sqrt{|V|})$ whereas the maximal packing has size exactly one (see \Cref{fig:network_large_gap} for $k=4$). We can also show that this instance exhibits an asymptotically optimal gap:

\begin{theorem} \label{thm:domset-v-packing}
For any graph $G$, $\OPT_{\LP} \leq \packnum(G) \cdot \sqrt{n}$ where $\OPT_{\LP}$ denote the value of the optimal solution to the LP in \eqref{eq:origlp}.
\end{theorem}

In other words, our upper bound based on the LP (\Cref{thm:lp-protocol}) and  our lower bound (\Cref{thm:tgdplb}) on the MSE has a gap of at most $O(\sqrt{n})$. To the best of our knowledge, the bound in \Cref{thm:domset-v-packing} was not known before; we give the full proof in \Cref{app:packing-v-domset}.

Note that the above instance also gives a gap of $\Omega(\sqrt{|V|})$ between the domination number and the packing number. Since it is known~\cite{Lovasz75} that $\gamma(G) \leq O(\log n) \cdot \OPT_{\LP}$, \Cref{thm:domset-v-packing} implies the following corollary:
\begin{corollary}
For any graph $G$, $\gamma(G) \leq \packnum(G) \cdot O(\sqrt{n} \cdot \log n)$.
\end{corollary}

That is, the above gap instance is tight up to a logarithmic factor. Furthermore, this also means that our upper bound based on the dominating set (\Cref{thm:domset-protocol}) and our lower bound (\Cref{thm:tgdplb}) on the MSE has a gap of at most $O(\sqrt{n})$.
\newcommand{\oT}{\overline{T}}

\section{Robust Trust Graph Differential Privacy}\label{ref:robust_trust_graph_DP}


In the previous section, we assumed  that each user $u$ trusts all of their neighbors $N(u)$. Although this is certainly a reasonable assumption, it might pose a security risk. For example, if one of the neighbors of $u$ is compromised, then $u$'s data might be leaked as the model offers no protection with respect to the view of $u$'s neighbors. Indeed, in the dominating set protocol (\Cref{thm:domset-protocol}), the user sends their raw data to one of their neighbors; if this neighbor is compromised, then the user's data is leaked in the clear. To mitigate this, we propose a revised trust graph DP model that is more robust to such leakage. In particular, for each user $u$, the DP protection remains as long as at most $t_u$ of their neighbors are compromised, where $t_u$ is some predefined number. This is formalized below.

\begin{definition} (Robust Trust Graph DP) \label{def:robust-graph-dp}
Let $G = (V, E)$ and $\bt = (t_v)_{v \in V} \in \Z_{\geq 0}^V$. A protocol $P$ satisfies \emph{$(\epsilon,\delta, G, \bt)$-Robust Trust Graph DP ($(\epsilon, \delta, G, \bt)$-RTGDP)} if for each vertex $v \in V$ and every set $T \subseteq N(v)$ of size at most $t_v$, $\VIEW^{V \setminus (N[v] \setminus T)}_P(\bx)$ satisfies $(\eps,\delta)$-DP with respect to the input $x_v$ for all values of $\bx_{-v}$. I.e., for all pairs $x_v, x'_v \in \mathcal{X}$, all values of $\bx_{-v}$, and all subsets $S \subseteq \cO$,
\[\Pr[\VIEW^{V \setminus (N[v] \setminus T)}_P(x_v, \bx_{-v}) \in S]\leq e^\eps \Pr[\VIEW^{V \setminus (N[v] \setminus T)}_P(x'_v, \bx_{-v}) \in S] + \delta.\]
\end{definition}

\subsection{Integer Aggregation Protocol}

We start by giving an integer aggregation protocol that is again based on an LP. We  adapt the LP in \eqref{eq:origlp} by imposing a stricter constraint to ensure DP guarantees even when up to $t_v$ of $v$'s neighbor are compromised. 
This results in the following LP where the only difference compared to \eqref{eq:origlp} is the stricter first constraint.\footnote{$\binom{S}{\leq t}$ denotes the collection of all subsets of $S$ of size $\leq t$.} 
Note that when $\bt = 0$, the two LPs coincide. 
\begin{equation} \label{eq:lp-robust}
\min\;\; \sum_{u \in V} y_u \qquad \text{s.t.}\;\; \sum_{u \in (N[v] \setminus T)} y_u \geq 1 \quad\forall v \in V, \;T \in \binom{N(v)}{\leq t_v}; \quad 0 \leq y_u \leq 1 \quad \forall u \in V. 
\end{equation}
While \eqref{eq:lp-robust} can have exponential size due to the presence of \emph{``$\forall T \in \binom{N(v)}{\leq t_v}$''} in the constraint, it can still be solved in polynomial time because an efficient separation oracle exists for the first constraint.
Specifically, for each $v \in V$, we can check the first constraint by letting $T$ be the $t_v$ maximum values and check the inequality for just that $T$ (instead of enumerating over all $T \in \binom{N(v)}{\leq t_v}$).
%
%
From the above LP, we can derive an algorithm that uses exactly the same protocol as in \Cref{thm:lp-protocol}.

\begin{theorem}\label{thm:lp-protocol-robust}
There is an $(\eps, G, \bt)$-RTGDP protocol for the aggregation problem with MSE at most $2\Delta^2\cdot\optlp^{\bt}/\eps^2$, where  $\optlp^{\bt}$ denotes the value of the optimal solution to the LP in \eqref{eq:lp-robust}.
\end{theorem}
%


\subsection{Lower Bound}

We define a $(2, \bt)$-robust packing of $G = (V, E)$ as a pair $U \subseteq V$ and $(T_u)_{u \in U}$ such that (i) $N[u] \setminus T_u$ are disjoint for all $u \in U$ and (ii) $T_u \in \binom{N(u)}{\leq t_u}$ for all $u \in U$. The size of the robust packing is $|U|$. Let $\optind^{\bt}$ denote the largest size of a $(2, \bt)$-robust packing of $G$.

Note that, when $\bt = \bzero$, $(2, \bt)$-robust packing coincides with the standard notion of packing we used in the previous section. We prove a lower bound for integer aggregation in the $(\eps, G, \bt)$-RTGDP model that grows with the size of the maximum $(2, \bt)$-robust packing:

\begin{theorem} \label{thm:robust-lb-generic}
For any $\eps \leq O(1)$ and $\bt \in \bN^{V}$, any $(\eps, G, \bt)$-RTGDP protocol for integer aggregation must incur MSE at least $\Omega(\Delta^2 \cdot \optind^{\bt})$.\footnote{Again, this theorem is shown via a reduction to the LDP model; see~\Cref{app:proofs} for the proof.
} 
\end{theorem}


%

\subsection{Bi-criteria Tightness of the Bounds}

Although we are not aware in general how large the gap between our upper (\Cref{thm:lp-protocol-robust}) and lower bounds (\Cref{thm:robust-lb-generic}) are, we can show the following \emph{bi-criteria} result, that the upper bound is not much larger than the lower bound when we increase $\bt$ slightly. This is stated and proved below. 

We write $\lceil \alpha \cdot \bt \rceil$ for some $\alpha > 0$ as a shorthand for the vector $(\lceil \alpha \cdot t_u\rceil)_{u \in U}$. Furthermore, let $\bdeg_U$ denote the vector of degrees of the vertices, i.e., $(\deg(u))_{u \in U}$.

\begin{theorem} \label{thm:robust-bicriteria-gap}
For any $\alpha \in (0, 1)$, $\optind^{\bt + \lceil \alpha \cdot \bdeg_U\rceil} \geq \frac{\alpha}{8} \cdot \optlp^{\bt}$, where $\optind^{\bt}, \optlp^{\bt}$ are as defined in \Cref{thm:lp-protocol-robust,thm:robust-lb-generic}.
\end{theorem}

To show this, we consider the dual of LP in \eqref{eq:lp-robust}, which turns out to be a relaxation for $(2, \bt)$-robust packing where there is a variable $w_{v, T} \in [0, 1]$ for all $v \in V, T \subseteq \binom{N(v)}{\leq t_v}$ representing whether $(v, T)$ should be included in the robust packing. To turn such a fractional solution to an integral one, we employ \emph{randomized rounding}---a standard technique in approximation algorithms (e.g., \cite[Chapter 5]{WS11}). More precisely, we include $(v, T)$ in our solution with probability proportional to $w_{v, T}$. Unfortunately, this does not work yet as the produced solution may not be a $(2, \bt)$-robust packing, i.e., it might contain $w_{v, T}$ and $w_{v', T'}$ such that $(N[v] \setminus T)$ and $(N[v'] \setminus T')$ are not disjoint. Due to this, we need to apply a correction procedure on top of this randomized solution. Roughly speaking, we try to enlarge $T$ until we are sure that such an intersection is avoided. This is indeed the reason why we need the slight increase in $\bt$. However, even with this increase, we still have to be careful as sometimes $T$ might become too large, i.e., larger than $t_v + \lceil\alpha \cdot \deg(v)\rceil$. We deal with this by simply removing such a pair $(v, T)$ from the solution. A careful analysis shows that (in expectation) only a small fraction of the solution will get removed this way; see~\Cref{app:proofs}.
\section{Machine Learning with Trust Graph DP}\label{app:ML}

\newcommand{\round}{\mathrm{round}}
\newcommand{\cN}{\mathcal{N}}

While the main body of our work focuses on integer aggregation, it is a primitive on which we can build many more complex algorithms. First of all, we can easily use it to perform \emph{real number aggregation} by re-scaling and discretization~\cite{BBGN20}. In particular, suppose that each user now has $x_i \in [0, 1]$. They can pick $\Delta \in \bN$ and perform integer aggregation on $y_i = \round(\Delta x_i)$ where $\round(x)$ randomly round $x$ to either $\lfloor x \rfloor$ or $1 + \lfloor x \rfloor$ with probabilities $1 - (x - \lfloor x \rfloor)$ and $x - \lfloor x \rfloor$ respectively. Once we have run the integer aggregation protocol, the answer is scaled by a factor of $\frac{1}{\Delta}$. If our integer aggregation protocol has MSE $\Delta^2 \xi^2$, then this results in a real aggregation protocol with error $\xi^2 + \frac{n}{4\Delta^2}$~\cite{BBGN20}. Picking $\Delta$ to be sufficiently large (e.g., $\omega(\sqrt{n})$), the second term becomes negligible.

Real number aggregation allows us to perform statistical queries (SQ)~\cite{Kearns98}. While a simple family, statistical queries have wide variety of applications in learning theory. One specific work we wish to highlight is that of \citep{FeldmanGV17}, who showed that convex optimization problems can be solved using statistical queries. As a result, we can apply their algorithm to our setting and obtain convex optimization algorithms with Trust Graph DP. We also remark that statistical queries are also useful for (non-ML) data analytic tasks; e.g.,~\citep{FeldmanGV17} provides SQ-based (vector) mean estimation, and other statistics such as quantiles are also known to be computable via SQs~\cite{Feldman17}.

\subsection{Vector Summation with Trust Graph DP}

Although SQ-based algorithms can be used in our model, we will sketch a more direct algorithm for the task of \emph{vector summation} with Trust Graph DP. This is not only more efficient but also provide better error guarantees for subsequent tasks such as convex empirical risk minimization (ERM).

In the \emph{vector summation} (with $\ell_2$-norm bound) problem, each user input $x_i$ is a vector in $\R^d$ such that $\|x_i\|_2 \leq \Delta$, where $\Delta$ is a norm bound known to the algorithm. The goal is again to compute an estimate $\ta$ to the sum $a = \sum_{i \in [n]} x_i$. The $\ell_2^2$-error is defined as $\E[\|\ta - a\|_2^2]$. Furthermore, we say that the estimator is unbiased if $\E[\ta] = a$.

It will be easiest to state the algorithms in terms of \emph{zero-concentrated differential privacy (zCDP)}~\cite{BunS16,DworkR16}. To do so, we first define $\alpha$-Renyi divergence for $\alpha > 1$ between distributions $\cP, \cP'$ for two distributions $\cP, \cP'$ to be $D_{\alpha}(\cP ~\|~ \cP') :=  \frac{1}{\alpha - 1} \ln\left(\E_{x \sim \cP}\left[\left(\frac{\cP(x)}{\cP'(x)}\right)^{\alpha - 1}\right]\right)$. We note that $\lim_{\alpha \to \infty} D_\alpha(\cP ~\|~ \cP')$ is indeed equal to $D_{\infty}(\cP ~\|~ \cP')$ that we defined in \Cref{sec:prelim}.

zCDP can now be defined as follows.

\begin{definition}
(zCDP;~\citep{BunS16})
A randomized mechanism $M: \mathcal{X}^n \to \cO$ is \emph{$\rho$-zero concentrated DP ($\rho$-zCDP)} if for all pairs $\mathbf{x}, \mathbf{x}' \in \mathcal{X}^n$ of datasets that differ only in the data of a single user and all $\alpha > 1$,
\begin{equation}
    D_\alpha\left(M(\mathbf{x}) ~\|~ M(\mathbf{x}')\right)\leq \alpha\rho.
\end{equation}
\end{definition}

ZCDP can be easily converted to DP:
\begin{lemma}[\citep{BunS16}] \label{lem:zcdp-to-dp}
For any $\rho > 0, \delta \in (0, 1/2)$, any $\rho$-zCDP algorithm is $(\rho + 2\sqrt{\rho \cdot \ln(1/\delta)}, \delta)$-DP.
\end{lemma}

It also has a simple composition theorem:
\begin{lemma}[\citep{BunS16}] \label{lem:cdp-comp}
Let $M$ be a mechanism that just runs subroutines that are $\rho_1$-zCDP, \dots, $\rho_m$-zCDP. Then, $M$ is $(\rho_1 + \cdots + \rho_m)$-zCDP. 
\end{lemma}

Trust Graph zCDP is simply zCDP on the view of the non-neighbors:
\begin{definition}[Trust Graph zCDP]
Let $G = (V, E)$.  A protocol $P$ satisfies \emph{$(\rho, G)$-Trust Graph zCDP ($(\rho, G)$-TGzCDP)} if for each vertex $v \in V$, $\VIEW^{V \setminus N[v]}_P(\bx)$ satisfies $\rho$-zCDP with respect to the input $x_v$ for all values of $\bx_{-v}$. I.e., for all pairs $x_v, x'_v \in \mathcal{X}$, all values of $\bx_{-v}$ and all $\alpha > 1$,
$$D_{\alpha}\left(\VIEW^{V \setminus N[v]}_P(x_v, \bx_{-v}) ~\middle\|~ \VIEW^{V \setminus N[v]}_P(x'_v, \bx_{-v})\right) \leq \alpha \rho.$$
\end{definition}

We can now state the algorithm for vector summation, which is similar to the algorithm in \Cref{thm:domset-protocol}.

\begin{theorem} \label{thm:domset-protocol-vecsum}
There is an $(\rho, G)$-TGzCDP mechanism for the aggregation problem with $\ell_2^2$-error $2 d \Delta^2 |T| / \rho$, where $T$ is any dominating set of $G$. Furthermore, the estimate is unbiased.
\end{theorem}

\begin{proof}
Let $\sigma = \Delta\sqrt{\frac{1}{2\rho}}$. The protocol works as follows:
\begin{itemize}
\item First, each user $v \in V$ picks an arbitrary vertex $u_v \in T \cap N[v]$.
\item Each user $u \in T$ broadcasts the sum of all vectors it receives together with a noise drawn from $\cN(0, \sigma^2 I_d)$. 
More formally, the user broadcasts $a_u = \sum_{v \in V \atop u_v = u} x_v + z_u$, where $z_u \sim \cN(0, \sigma^2 I_d)$.
\item Finally, the estimate is $\ta = \sum_{u \in T} a_u$.
\end{itemize}

\paragraph*{Privacy Analysis.}
Consider any $v \in V$ and $\bx_{-v} \in \cX^{V \setminus \{v\}}$. 
Similar to the proof of \Cref{thm:domset-protocol}, the view is a post-processing of $z_u + x$. Thus, for any $x_v, x'_v \in \cX$, we have
\begin{align*}
D_{\alpha}\left(\VIEW^{V \setminus N[v]}_P(x_v)~\middle\|~ \VIEW^{V \setminus N[v]}_P(x'_v)\right)
\leq D_{\alpha}(z_{u_v} + x_v ~\|~ z_{u_v} + x'_v) \leq \alpha \cdot \rho,
\end{align*}
where the second inequality follows from the zCDP guarantee of the Gaussian mechanism (e.g., Proposition 16 of \citep{BunS16}).
Thus, the protocol satisfies $(\rho, G)$-TGzCDP as desired.

\paragraph*{Utility Analysis.}
It is clear that the estimate is unbiased.
The $\ell_2^2$-error is
\begin{align*}
\E\left[\left(\ta - a\right)^2\right] 
= \sum_{u \in T} \E[\|z_u\|^2]
= |T| \cdot (\sigma^2 d) = 2d\Delta^2|T|/\rho. & \qedhere
\end{align*}
\end{proof}

Using \Cref{lem:zcdp-to-dp}, we immediately get the following corollary. For $|T| = \Theta(1)$, this guarantee (asymptotically) matches the known lower bound in central DP (see, e.g., \cite{BassilyST14}), while for $|T| = \Theta(n)$ , this guarantee nearly matches the known lower bound in local DP \cite{AsiFT22}.
\begin{corollary} \label{cor:domset-protocol-vecsum}
For any $\eps < O(\log(1/\delta))$, there is an $(\eps, \delta, G)$-TGDP mechanism for the aggregation problem with $\ell_2^2$-error $O\left(\frac{d \Delta^2 |T| \log(1/\delta)}{\eps
^2}\right)$, where $T$ is any dominating set of $G$. Furthermore, the estimate is unbiased.
\end{corollary}

\subsection{From Vector Summation to Convex Optimizaiton}

Using the above vector summation protocol, we can immediately implement several DP ML algorithms in the literature, such as the DP-SGD algorithm~\cite{AbadiCGMMT016}. We can also obtain formal guarantee for convex optimization problems from these algorithms. For instance, let us consider the convex ERM problem~\cite{BassilyST14}. Here there is a loss function $\ell: \cW \times \cX \to \R$ which is $L$-Lipschitz on the first parameter and suppose that the diameter of $\cW$ is at most $R$. The goal is to minimize the empirical loss $\cL(w, \bx) := \frac{1}{n} \sum_{v \in V} \ell(w, x_v)$. Using the above vector summation algorithm, we can arrive at the following:

\begin{theorem}
For $0 < \rho \leq O(1)$, there is an $(\rho, G)$-TGzCDP mechanism for the convex ERM problem with expected excess risk $O\left(\frac{RL \sqrt{d |T| / \rho}}{n}\right)$, where $T$ is any dominating set of $G$.
\end{theorem}

\begin{proof}[Proof Sketch]
The algorithm uses the (stochastic) mirror descent (see, e.g.,  \citep[Section 6.1]{Bubeck15}) over $n^2$ steps. In each step, the gradient is computed by running our $(\rho', G)$-TGzCDP vector summation protocol with $\rho' = \rho / n^2, \Delta = G$ and scale the answer by a factor of $\frac{1}{n}$. The composition theorem (\Cref{lem:cdp-comp}) implies that the algorithm is $\rho$-TGzCDP. The excess risk guarantee follows from standard analysis of stochastic mirror descent (e.g., \citep[Theorem 6.1]{Bubeck15}).
\end{proof}

Again, using \Cref{lem:zcdp-to-dp}, we immediately get the following corollary for $(\eps, \delta)$-DP. For $|T| = \Theta(1)$, this guarantee (asymptotically) matches the known lower bound in central DP \cite{BassilyST14}.
\begin{corollary}
For any $\eps < O(\log(1/\delta))$, there is an $(\eps, \delta, G)$-TGDP mechanism for the convex ERM problem with expected excess risk $O\left(\frac{RL \sqrt{d |T| \log(1/\delta)}}{\eps n}\right)$, where $T$ is any dominating set of $G$.
\end{corollary}

\section{Experiments}\label{app:exp}

We analyze the theoretically proven error bounds on real datasets. We focus  on: \textit{(i)} the \textit{improvements} over the error that come from applying the notion of TGDP compared to LDP, and \textit{(ii)} the \textit{trade-off} between robustness and efficiency when applying RTGDP. Experiments on real networks show that the LP protocol (\Cref{thm:lp-protocol}) for achieving TGDP yields a significant reduction in the error bound compared to LDP.  Furthermore, the protocol for achieving RTGDP (\Cref{thm:lp-protocol-robust}) still has markedly lower error than LDP even when a large number of neighbors could be compromised.  

\subsection{Datasets and implementation.} We analyzed nine network datasets available from the Stanford Large Network Dataset Collection~\citep{snapnets}: three email communication, two Bitcoin trust network, and four social network datasets. 
Individual details about each dataset are given in~\Cref{app:exp:data}.
Our analysis of the MSE of the proposed protocols focuses on computing $\OPT_{\LP}$ and $\OPT^{\mathbf{t}}_{\LP}$ for each graph. 
These were computed using \texttt{cvxpy} \citep{diamond2016cvxpy,agrawal2018rewriting}, and code will be made publicly available on publication.

\subsection{Error of TGDP vs. LDP.}
We compare the upper bound on the MSE using the mechanism described in~\Cref{thm:lp-protocol}  to that of standard LDP mechanisms~(\Cref{def:int-ldp}).  Specifically, the ($\varepsilon, G$)-TGDP mechanism proposed in Theorem \ref{thm:lp-protocol} achieves an MSE at most $2\Delta^2 \cdot \OPT_{\LP}/\varepsilon^2$ whereas LDP with a local version of the Laplace mechanism\footnote{There are other LDP mechanisms with slightly different error. We discuss this more in Appendix~\ref{app:local-mech-agg}.} achieves an error of $2 \Delta^2 n / \varepsilon^2$.  We call the ratio of these quantities (i.e., $\frac{\OPT_{\LP}}{n}$) as the \textit{TGDP error ratio} for a given graph. 

Table \ref{tab:tgdp_bounds} 
shows the TGDP error ratio for all datasets. On the evaluated graph, the TGDP error ratio ranged between $0.002$ and $0.207$, showing that the TGDP protocol appreciably reduces MSE.  In other words, accounting for trust between users can significantly reduce error.

\subsection{Gap between upper and lower error bounds for TGDP.} We also present a lower bound on the error of a TGDP mechanism for the datasets in Table \ref{tab:tgdp_bounds} by reporting the size of a maximal packing, which lower bounds the term $\optind$ from Theorem \ref{thm:tgdplb}. A maximal packing is a set $U \subseteq V$ such that if any other vertex $v \in V \setminus U$ were added, the result would no longer be a packing. The gap between upper and lower bound is fairly small for all datasets.


\begin{table}[!ht]
    \centering
    \begin{tabular}{lrrrr}
\toprule
Dataset & $n$ & $\OPT_{\LP}$ & $\frac{\OPT_{\LP}}{n}$ & Maximal \\
& & & & packing \\
\midrule
EU Emails (Core)  & $1,005$ & $111.97$ & $0.111$ & $103$\\
Bitcoin (Alpha) & $3,783$ & $686$ & $0.181$ & $480$\\
Facebook & $4,039$ & $10$ & $0.002$ & $10$\\
Bitcoin (OTC) & $5,881$ & $1126$ & $0.191$ & $691$ \\
Enron Emails   & $36,692$ & $3,060.66$ & $0.083$ & $2,784$ \\
GitHub & $37,700$ & $4,538$ & $0.120$ & $3,910$ \\
Epinions & $75,879$ & $15,734$ & $0.207$ & $13,832$ \\
Twitter & $81,306$ & $961$ & $0.012$ & $843$ \\
EU Emails (All) & $265,214$ & $18,074.40$ & $0.068$ & $18,007$ \\
\bottomrule
\end{tabular}
    \caption{Upper and lower bounds on error of a TGDP protocol for 9 network datasets. $\OPT_{\LP}/n$ is the error ratio, which is exactly the ratio between the error of the proposed TGDP protocol and the local Laplace mechanism. Across all datasets, this ratio indicates significant gains in efficiency for TGDP protocols. The last column is the size of a maximal packing, which lower bounds the term $\optind$ from Theorem \ref{thm:tgdplb}. The gap between upper and lower bound is fairly small for all datasets.}
    \label{tab:tgdp_bounds}
\end{table}

\subsection{Error of RTGDP.}
The RTGDP definition is more conservative than TGDP, especially as the number of untrusted neighbors per node increases. 
To evaluate the trade-off between robustness and error, we compute $\OPT^{\mathbf{t}}_{\LP}$ for varying degrees of mistrust $\mathbf{t}$ for four graphs of size $n < 10,000$. These include one email communication graph, two Bitcoin trust graphs, and one social network. For each graph, we vary the degree of mistrust $\mathbf{t}$ as follows:
for each $\alpha$, we construct $\bt_{\alpha} = \lceil \alpha \cdot \bdeg_V \rceil = (\lceil \alpha \cdot \deg(v) \rceil)_{v \in V} $.
Thus, $\alpha$ represents the proportion of neighbors that are untrusted for each node. We compute $\OPT^{\mathbf{t}_{\alpha}}_{\LP}$ for various $\alpha \in [0,1]$.
As before, we study the \textit{$\mathbf{t}_{\alpha}$-RTGDP error ratio}, $\frac{\OPT^{\mathbf{t}_{\alpha}}_{\LP}}{n}$.

Figure \ref{fig:rtgdp} shows the $\mathbf{t}_{\alpha}$-RTGDP error ratios for varying levels of mistrust $\alpha$. The initial jump in the error ratio at $\alpha=0.1$ comes from the ceiling rounding, which affects all nodes with a single neighbor. Notably, for most datasets, an error ratio of less than $0.6$ can still be achieved for $\alpha$ up to $0.5$. This suggests that the RTGDP notion may be a  practical way to interpolate between LDP and TGDP, as considerable error reductions can still occur for RTGDP for $\alpha < 0.5$.

\begin{figure}[!ht]
    \centering
    \begin{tblr}{p{0.8cm}cc}
        & $\mathbf{t}_{\alpha}$-RTGDP error ratios for varying levels of mistrust $\alpha$ \\
        \SetCell[r=2]{}$\frac{\OPT^{\mathbf{t}_{\alpha}}_{\LP}}{n}$ & \includegraphics[width=0.5\columnwidth]{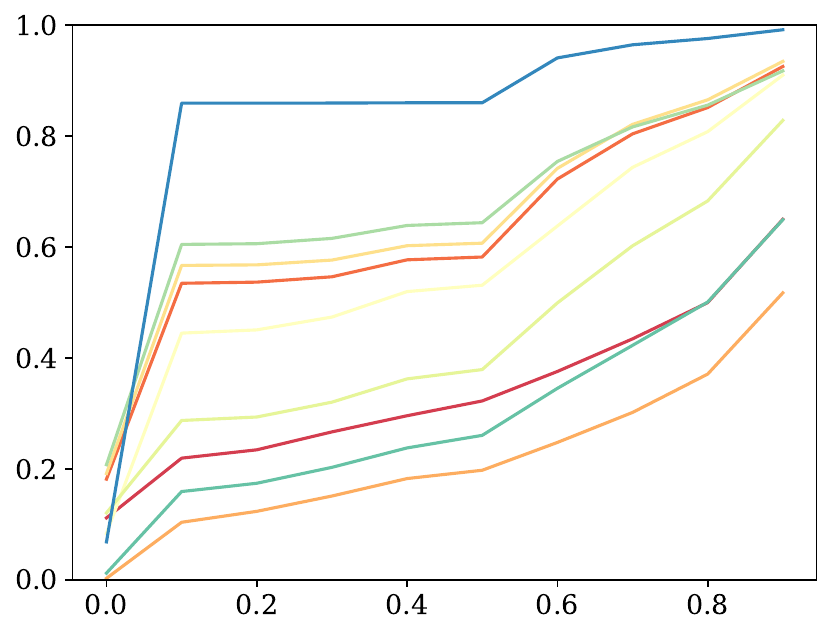} & \includegraphics[width=0.3\columnwidth]{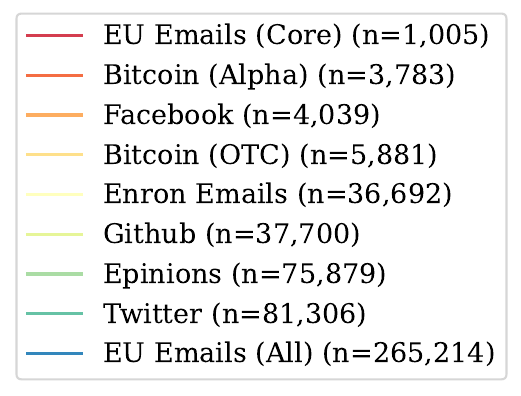} \\
        & Proportion of untrusted neighbors $\alpha$ \\
    \end{tblr}
    \caption{$\mathbf{t}_{\alpha}$-RTGDP error ratios for varying levels of mistrust $\alpha$. Even for relatively high levels of mistrust $\alpha$ up to $0.5$, RTGDP can yield a substantial improvement in error relative to LDP for most datasets. Note that for $\alpha = 1$, $\mathbf{t}_{\alpha}$-RTGDP is equivalent to LDP.}
    \label{fig:rtgdp}
\end{figure}

\section{Conclusion and Future Directions}\label{sec:conc_fut_dir}

We have proposed a new model of privacy given a graph of trust relationships between users. Our model generalizes central and local DP, and further captures intermediate trust structures such as social networks or multiple curators. 
A significant open theoretical problem is to close the gap between the upper and lower bounds for TGDP, though our experiments suggest that this gap may be small in practice. 


\bibliographystyle{alpha}
\bibliography{main_arxiv}

\newcommand{\etalchar}[1]{$^{#1}$}
\begin{thebibliography}{GKM{\etalchar{+}}21}

\bibitem[Abo18]{abowd2018us}
John~M Abowd.
\newblock The {US Census Bureau} adopts differential privacy.
\newblock In {\em KDD}, pages 2867--2867, 2018.

\bibitem[ACG{\etalchar{+}}16]{AbadiCGMMT016}
Mart{\'{\i}}n Abadi, Andy Chu, Ian~J. Goodfellow, H.~Brendan McMahan, Ilya
  Mironov, Kunal Talwar, and Li~Zhang.
\newblock Deep learning with differential privacy.
\newblock In {\em CCS}, pages 308--318, 2016.

\bibitem[AFT22]{AsiFT22}
Hilal Asi, Vitaly Feldman, and Kunal Talwar.
\newblock Optimal algorithms for mean estimation under local differential
  privacy.
\newblock In {\em ICML}, pages 1046--1056, 2022.

\bibitem[AVDB18]{agrawal2018rewriting}
Akshay Agrawal, Robin Verschueren, Steven Diamond, and Stephen Boyd.
\newblock A rewriting system for convex optimization problems.
\newblock {\em Journal of Control and Decision}, 5(1):42--60, 2018.

\bibitem[BBG{\etalchar{+}}20]{bell2020secure}
James~Henry Bell, Kallista~A Bonawitz, Adri{\`a} Gasc{\'o}n, Tancr{\`e}de
  Lepoint, and Mariana Raykova.
\newblock Secure single-server aggregation with (poly) logarithmic overhead.
\newblock In {\em CCS}, pages 1253--1269, 2020.

\bibitem[BBGN19]{balle2019privacy}
Borja Balle, James Bell, Adri{\`a} Gasc{\'o}n, and Kobbi Nissim.
\newblock The privacy blanket of the shuffle model.
\newblock In {\em CRYPTO}, pages 638--667, 2019.

\bibitem[BBGN20]{BBGN20}
Borja Balle, James Bell, Adri{\`{a}} Gasc{\'{o}}n, and Kobbi Nissim.
\newblock Private summation in the multi-message shuffle model.
\newblock In {\em CCS}, pages 657--676, 2020.

\bibitem[BEM{\etalchar{+}}17]{bittau17}
Andrea Bittau, {\'{U}}lfar Erlingsson, Petros Maniatis, Ilya Mironov, Ananth
  Raghunathan, David Lie, Mitch Rudominer, Ushasree Kode, Julien Tinn{\'{e}}s,
  and Bernhard Seefeld.
\newblock Prochlo: Strong privacy for analytics in the crowd.
\newblock In {\em SOSP}, pages 441--459, 2017.

\bibitem[BHvV09]{Burger}
Alewyn~P. Burger, Michael~A. Henning, and Jan~H. van Vuuren.
\newblock On the ratios between packing and domination parameters of a graph.
\newblock {\em Discrete Mathematics}, 309:2473--2478, 2009.

\bibitem[BIK{\etalchar{+}}17]{bonawitz2017practical}
Keith Bonawitz, Vladimir Ivanov, Ben Kreuter, Antonio Marcedone, H~Brendan
  McMahan, Sarvar Patel, Daniel Ramage, Aaron Segal, and Karn Seth.
\newblock Practical secure aggregation for privacy-preserving machine learning.
\newblock In {\em CCS}, pages 1175--1191, 2017.

\bibitem[BNO08]{BeimelNO08}
Amos Beimel, Kobbi Nissim, and Eran Omri.
\newblock Distributed private data analysis: Simultaneously solving how and
  what.
\newblock In {\em CRYPTO}, pages 451--468, 2008.

\bibitem[BS16]{BunS16}
Mark Bun and Thomas Steinke.
\newblock Concentrated differential privacy: Simplifications, extensions, and
  lower bounds.
\newblock In {\em TCC}, pages 635--658, 2016.

\bibitem[BST14]{BassilyST14}
Raef Bassily, Adam~D. Smith, and Abhradeep Thakurta.
\newblock Private empirical risk minimization: Efficient algorithms and tight
  error bounds.
\newblock In {\em FOCS}, pages 464--473, 2014.

\bibitem[Bub15]{Bubeck15}
S{\'{e}}bastien Bubeck.
\newblock Convex optimization: Algorithms and complexity.
\newblock {\em Found. Trends Mach. Learn.}, 8(3-4):231--357, 2015.

\bibitem[CB22]{cyffers2022privacy}
Edwige Cyffers and Aur{\'e}lien Bellet.
\newblock Privacy amplification by decentralization.
\newblock In {\em AISTATS}, pages 5334--5353, 2022.

\bibitem[CIY20]{cho2020contact}
Hyunghoon Cho, Daphne Ippolito, and Yun~William Yu.
\newblock Contact tracing mobile apps for {COVID-19}: Privacy considerations
  and related trade-offs.
\newblock {\em arXiv preprint arXiv:2003.11511}, 2020.

\bibitem[CSS12]{ChanSS12}
T.{-}H.~Hubert Chan, Elaine Shi, and Dawn Song.
\newblock Optimal lower bound for differentially private multi-party
  aggregation.
\newblock In {\em ESA}, pages 277--288, 2012.

\bibitem[CSU{\etalchar{+}}19]{CheuSUZZ19}
Albert Cheu, Adam~D. Smith, Jonathan Ullman, David Zeber, and Maxim Zhilyaev.
\newblock Distributed differential privacy via shuffling.
\newblock In {\em EUROCRYPT}, pages 375--403, 2019.

\bibitem[CY23]{cheu2023necessary}
Albert Cheu and Chao Yan.
\newblock Necessary conditions in multi-server differential privacy.
\newblock In {\em ITCS}, 2023.

\bibitem[DB16]{diamond2016cvxpy}
Steven Diamond and Stephen Boyd.
\newblock {CVXPY}: {A} {P}ython-embedded modeling language for convex
  optimization.
\newblock {\em JMLR}, 17(83):1--5, 2016.

\bibitem[DKBS15]{dong2015differential}
Roy Dong, Walid Krichene, Alexandre~M Bayen, and S~Shankar Sastry.
\newblock Differential privacy of populations in routing games.
\newblock In {\em CDC}, pages 2798--2803, 2015.

\bibitem[DKM{\etalchar{+}}06]{dwork2006our}
Cynthia Dwork, Krishnaram Kenthapadi, Frank McSherry, Ilya Mironov, and Moni
  Naor.
\newblock Our data, ourselves: Privacy via distributed noise generation.
\newblock In {\em EUROCRYPT}, pages 486--503, 2006.

\bibitem[DKY17]{ding2017collecting}
Bolin Ding, Janardhan Kulkarni, and Sergey Yekhanin.
\newblock Collecting telemetry data privately.
\newblock {\em NIPS}, 30, 2017.

\bibitem[DMNS06]{dwork2006calibrating}
Cynthia Dwork, Frank McSherry, Kobbi Nissim, and Adam~D. Smith.
\newblock Calibrating noise to sensitivity in private data analysis.
\newblock In {\em TCC}, pages 265--284, 2006.

\bibitem[DR{\etalchar{+}}14]{dwork2014algorithmic}
Cynthia Dwork, Aaron Roth, et~al.
\newblock The algorithmic foundations of differential privacy.
\newblock {\em Foundations and Trends{\textregistered} in Theoretical Computer
  Science}, 9(3--4):211--407, 2014.

\bibitem[DR16]{DworkR16}
Cynthia Dwork and Guy~N. Rothblum.
\newblock Concentrated differential privacy.
\newblock {\em CoRR}, abs/1603.01887, 2016.

\bibitem[EFM{\etalchar{+}}19]{erlingsson2019amplification}
{\'U}lfar Erlingsson, Vitaly Feldman, Ilya Mironov, Ananth Raghunathan, Kunal
  Talwar, and Abhradeep Thakurta.
\newblock Amplification by shuffling: From local to central differential
  privacy via anonymity.
\newblock In {\em SODA}, pages 2468--2479, 2019.

\bibitem[EGS03]{evfimievski2003limiting}
Alexandre Evfimievski, Johannes Gehrke, and Ramakrishnan Srikant.
\newblock Limiting privacy breaches in privacy preserving data mining.
\newblock In {\em PODS}, pages 211--222, 2003.

\bibitem[Fei98]{Feige98}
Uriel Feige.
\newblock A threshold of ln \emph{n} for approximating set cover.
\newblock {\em JACM}, 45(4):634--652, 1998.

\bibitem[Fel17]{Feldman17}
Vitaly Feldman.
\newblock Dealing with range anxiety in mean estimation via statistical
  queries.
\newblock In {\em ALT}, pages 629--640, 2017.

\bibitem[FGV17]{FeldmanGV17}
Vitaly Feldman, Crist{\'{o}}bal Guzm{\'{a}}n, and Santosh~S. Vempala.
\newblock Statistical query algorithms for mean vector estimation and
  stochastic convex optimization.
\newblock In {\em SODA}, pages 1265--1277, 2017.

\bibitem[GKM{\etalchar{+}}21]{ghazi2021differentially}
Badih Ghazi, Ravi Kumar, Pasin Manurangsi, Rasmus Pagh, and Amer Sinha.
\newblock Differentially private aggregation in the shuffle model: Almost
  central accuracy in almost a single message.
\newblock In {\em ICML}, pages 3692--3701, 2021.

\bibitem[GKMP20]{ghazi2020private}
Badih Ghazi, Ravi Kumar, Pasin Manurangsi, and Rasmus Pagh.
\newblock Private counting from anonymous messages: Near-optimal accuracy with
  vanishing communication overhead.
\newblock In {\em ICML}, pages 3505--3514, 2020.

\bibitem[GMPV20]{GMPV20}
Badih Ghazi, Pasin Manurangsi, Rasmus Pagh, and Ameya Velingker.
\newblock Private aggregation from fewer anonymous messages.
\newblock In {\em EUROCRYPT}, pages 798--827, 2020.

\bibitem[GMW19]{goldreich2019play}
Oded Goldreich, Silvio Micali, and Avi Wigderson.
\newblock How to play any mental game, or a completeness theorem for protocols
  with honest majority.
\newblock In {\em Providing Sound Foundations for Cryptography: On the Work of
  Shafi Goldwasser and Silvio Micali}, pages 307--328. 2019.

\bibitem[GRS12]{GhoshRS12}
Arpita Ghosh, Tim Roughgarden, and Mukund Sundararajan.
\newblock Universally utility-maximizing privacy mechanisms.
\newblock {\em SICOMP}, 41(6):1673--1693, 2012.

\bibitem[HKT00]{HalldorssonKT00}
Magn{\'{u}}s~M. Halld{\'{o}}rsson, Jan Kratochv{\'{\i}}l, and Jan~Arne Telle.
\newblock Independent sets with domination constraints.
\newblock {\em Discret. Appl. Math.}, 99(1-3):39--54, 2000.

\bibitem[IKOS06]{IKOS06}
Yuval Ishai, Eyal Kushilevitz, Rafail Ostrovsky, and Amit Sahai.
\newblock Cryptography from anonymity.
\newblock In {\em FOCS}, pages 239--248, 2006.

\bibitem[Kea98]{Kearns98}
Michael~J. Kearns.
\newblock Efficient noise-tolerant learning from statistical queries.
\newblock {\em JACM}, 45(6):983--1006, 1998.

\bibitem[KHM{\etalchar{+}}18]{kumar2018rev2}
Srijan Kumar, Bryan Hooi, Disha Makhija, Mohit Kumar, Christos Faloutsos, and
  VS~Subrahmanian.
\newblock Rev2: Fraudulent user prediction in rating platforms.
\newblock In {\em WSDM}, pages 333--341, 2018.

\bibitem[KLN{\etalchar{+}}11]{kasiviswanathan2011can}
Shiva~Prasad Kasiviswanathan, Homin~K Lee, Kobbi Nissim, Sofya Raskhodnikova,
  and Adam Smith.
\newblock What can we learn privately?
\newblock {\em SICOMP}, 40(3):793--826, 2011.

\bibitem[KSSF16]{kumar2016edge}
Srijan Kumar, Francesca Spezzano, VS~Subrahmanian, and Christos Faloutsos.
\newblock Edge weight prediction in weighted signed networks.
\newblock In {\em ICDM}, pages 221--230, 2016.

\bibitem[KY04]{klimt2004enron}
Bryan Klimt and Yiming Yang.
\newblock The {Enron} corpus: A new dataset for email classification research.
\newblock In {\em ECML}, pages 217--226, 2004.

\bibitem[LK14]{snapnets}
Jure Leskovec and Andrej Krevl.
\newblock {SNAP Datasets}: {Stanford} large network dataset collection.
\newblock \url{http://snap.stanford.edu/data}, June 2014.

\bibitem[LKF07]{leskovec2007graph}
Jure Leskovec, Jon Kleinberg, and Christos Faloutsos.
\newblock Graph evolution: Densification and shrinking diameters.
\newblock {\em TKDD}, 1(1), 2007.

\bibitem[LM12]{leskovec2012learning}
Jure Leskovec and Julian Mcauley.
\newblock Learning to discover social circles in ego networks.
\newblock {\em NIPS}, 25, 2012.

\bibitem[Lov75]{Lovasz75}
L{\'{a}}szl{\'{o}} Lov{\'{a}}sz.
\newblock On the ratio of optimal integral and fractional covers.
\newblock {\em Discret. Math.}, 13(4):383--390, 1975.

\bibitem[PCK20]{park2020information}
Sangchul Park, Gina~Jeehyun Choi, and Haksoo Ko.
\newblock Information technology-based tracing strategy in response to
  {COVID-19} in {South Korea}---privacy controversies.
\newblock {\em JAMA}, 323(21):2129--2130, 2020.

\bibitem[RAD03]{richardson2003trust}
Matthew Richardson, Rakesh Agrawal, and Pedro Domingos.
\newblock Trust management for the semantic web.
\newblock In {\em ISWC}, pages 351--368, 2003.

\bibitem[RAS19]{rozemberczki2019multiscale}
Benedek Rozemberczki, Carl Allen, and Rik Sarkar.
\newblock Multi-scale attributed node embedding.
\newblock {\em arXiv preprint arXiv:1909.13021}, 2019.

\bibitem[RE19]{tf-privacy}
Carey Radebaugh and Ulfar Erlingsson.
\newblock {Introducing TensorFlow Privacy: Learning with Differential Privacy
  for Training Data}, March 2019.
\newblock \url{blog.tensorflow.org}.

\bibitem[Sol02]{solove2002conceptualizing}
Daniel~J Solove.
\newblock Conceptualizing privacy.
\newblock {\em Calif. L. Rev.}, 90:1087, 2002.

\bibitem[Ste20]{steinke2020multi}
Thomas Steinke.
\newblock Multi-central differential privacy.
\newblock {\em arXiv preprint arXiv:2009.05401}, 2020.

\bibitem[TM20]{pytorch-privacy}
Davide Testuggine and Ilya Mironov.
\newblock {PyTorch Differential Privacy Series Part 1: DP-SGD Algorithm
  Explained}, August 2020.
\newblock \url{medium.com}.

\bibitem[Vad17]{vadhan2017complexity}
Salil Vadhan.
\newblock The complexity of differential privacy.
\newblock {\em Tutorials on the Foundations of Cryptography: Dedicated to Oded
  Goldreich}, pages 347--450, 2017.

\bibitem[War65]{warner1965randomized}
Stanley~L Warner.
\newblock Randomized response: A survey technique for eliminating evasive
  answer bias.
\newblock {\em JASA}, 60(309):63--69, 1965.

\bibitem[WS11]{WS11}
David~P. Williamson and David~B. Shmoys.
\newblock {\em The Design of Approximation Algorithms}.
\newblock Cambridge University Press, 2011.

\bibitem[Yao82]{yao1982protocols}
Andrew C-C Yao.
\newblock Protocols for secure computations.
\newblock In {\em FOCS}, pages 160--164, 1982.

\bibitem[Yao86]{yao1986generate}
Andrew C-C Yao.
\newblock How to generate and exchange secrets.
\newblock In {\em FOCS}, pages 162--167, 1986.

\end{thebibliography}


\appendix

\section{Example Graphs}

This section gives additional examples of trust graph structures to build intuition for the TGDP privacy model and for the theoretical results.

Figure \ref{fig:network_large_gap} gives an example of a graph with a significant gap between the domination number and the packing number, which importantly illustrates the gap between the upper and lower bounds on MSE under TGDP in Theorems \ref{thm:domset-protocol}, \ref{thm:lp-protocol}, and \ref{thm:tgdplb}. 

Figure \ref{fig:network_star} illustrates a simple example of a star graph in which the TGDP model captures central DP. 

\begin{figure}[!ht]
    \centering
    \begin{tikzpicture}
    \node[shape=circle,draw=black] (00) at (0,0) {};
    \node[shape=circle,draw=black] (01) at (0,1) {};
    \node[shape=circle,draw=black] (02) at (0,2) {};
    \node[shape=circle,draw=black] (03) at (0,3) {};
    \node[shape=circle,draw=black] (10) at (1,0) {};
    \node[shape=circle,draw=black] (11) at (1,1) {};
    \node[shape=circle,draw=black] (12) at (1,2) {};
    \node[shape=circle,draw=black] (13) at (1,3) {};
    \node[shape=circle,draw=black] (20) at (2,0) {};
    \node[shape=circle,draw=black] (21) at (2,1) {};
    \node[shape=circle,draw=black] (22) at (2,2) {};
    \node[shape=circle,draw=black] (23) at (2,3) {};
    \node[shape=circle,draw=black] (30) at (3,0) {};
    \node[shape=circle,draw=black] (31) at (3,1) {};
    \node[shape=circle,draw=black] (32) at (3,2) {};
    \node[shape=circle,draw=black] (33) at (3,3) {};

    \path [-] (00) edge node[left] {} (01);
    \path [-] (00) edge[bend left] node[left] {} (02);
    \path [-] (00) edge[bend left] node[left] {} (03);
    \path [-] (01) edge node[left] {} (02);
    \path [-] (02) edge node[left] {} (03);
    \path [-] (01) edge[bend left] node[left] {} (03);
    
    \path [-] (00) edge node[left] {} (10);
    \path [-] (00) edge[bend right] node[left] {} (20);
    \path [-] (00) edge[bend right] node[left] {} (30);
    \path [-] (10) edge node[left] {} (20);
    \path [-] (20) edge node[left] {} (30);
    \path [-] (10) edge[bend right] node[left] {} (30);

    \path [-] (10) edge node[left] {} (11);
    \path [-] (10) edge[bend left] node[left] {} (12);
    \path [-] (10) edge[bend left] node[left] {} (13);
    \path [-] (11) edge node[left] {} (12);
    \path [-] (12) edge node[left] {} (13);
    \path [-] (11) edge[bend left] node[left] {} (13);

    \path [-] (01) edge node[left] {} (11);
    \path [-] (01) edge[bend right] node[left] {} (21);
    \path [-] (01) edge[bend right] node[left] {} (31);
    \path [-] (11) edge node[left] {} (21);
    \path [-] (21) edge node[left] {} (31);
    \path [-] (11) edge[bend right] node[left] {} (31);

    \path [-] (20) edge node[left] {} (21);
    \path [-] (20) edge[bend right] node[left] {} (22);
    \path [-] (20) edge[bend right] node[left] {} (23);
    \path [-] (21) edge node[left] {} (22);
    \path [-] (22) edge node[left] {} (23);
    \path [-] (21) edge[bend right] node[left] {} (23);

    \path [-] (02) edge node[left] {} (12);
    \path [-] (02) edge[bend left] node[left] {} (22);
    \path [-] (02) edge[bend left] node[left] {} (32);
    \path [-] (12) edge node[left] {} (22);
    \path [-] (22) edge node[left] {} (32);
    \path [-] (12) edge[bend left] node[left] {} (32);

    \path [-] (30) edge node[left] {} (31);
    \path [-] (30) edge[bend right] node[left] {} (32);
    \path [-] (30) edge[bend right] node[left] {} (33);
    \path [-] (31) edge node[left] {} (32);
    \path [-] (32) edge node[left] {} (33);
    \path [-] (31) edge[bend right] node[left] {} (33);

    \path [-] (03) edge node[left] {} (13);
    \path [-] (03) edge[bend left] node[left] {} (23);
    \path [-] (03) edge[bend left] node[left] {} (33);
    \path [-] (13) edge node[left] {} (23);
    \path [-] (23) edge node[left] {} (33);
    \path [-] (13) edge[bend left] node[left] {} (33);
\end{tikzpicture}
    \caption{A graph with a gap between the domination number (4) and the packing number (1). The relaxed LP solution $\OPT_{\LP} = 16/7 \approx 2.285$.}
    \label{fig:network_large_gap}
\end{figure}
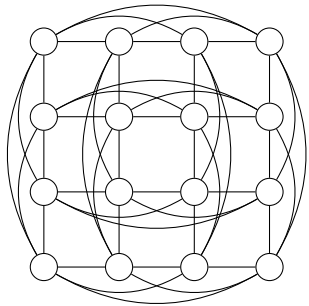

\begin{figure}[!ht]
    \centering
    \begin{tikzpicture}
    \node[shape=circle,draw=black] (00) at (0,0) {};
    \node[shape=circle,draw=black] (01) at (-0.5,1) {};
    \node[shape=circle,draw=black] (02) at (0,2) {};
    \node[shape=circle,draw=black] (10) at (1,-0.5) {};
    \node[shape=circle,draw=black] (11) at (1,1) {A};
    \node[shape=circle,draw=black] (12) at (1,2.5) {};
    \node[shape=circle,draw=black] (20) at (2,0) {};
    \node[shape=circle,draw=black] (21) at (2.5,1) {};
    \node[shape=circle,draw=black] (22) at (2,2) {};

    \path [-] (00) edge node[left] {} (11);
    \path [-] (01) edge node[left] {} (11);
    \path [-] (02) edge node[left] {} (11);
    \path [-] (10) edge node[left] {} (11);
    \path [-] (12) edge node[left] {} (11);
    \path [-] (20) edge node[left] {} (11);
    \path [-] (21) edge node[left] {} (11);
    \path [-] (22) edge node[left] {} (11);
    
\end{tikzpicture}
    \caption{A star graph representing a central DP setting in which all users trust user A as a central analyst. TGDP exactly matches central DP here in that TGDP would allow for user A to handle all data, while no individual user's data can be identified by communication across non-neighbors. The minimum dominating set size is simply 1.}
    \label{fig:network_star}
\end{figure}
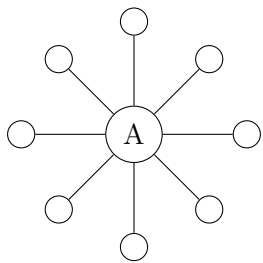

\section{Motivating Examples for Trust Graph DP}\label{app:motivating_examples}
The DP on trust graphs model is an intermediate setting between local and central DP. It allows formalizing a notion of (differential) privacy in cases where users have different privacy levels for different subsets of the population, as well as settings where  users can place different levels of trust on different digital services or sub-platforms. While one can think of different settings where trust graphs naturally arise, we next discuss a few example scenarios:

\paragraph*{Social graphs.} In a social network setting where a user’s posts are visible to their friends but hidden from strangers, computing analytics (e.g., a histogram of users’s check-ins in different geographic locations) or training a machine learning model (e.g., a generative model on the users’ posted photos) without entrusting a central curator naturally leads to peer-to-peer protocols that are trust graph DP with the trust graph being the social graph.

\paragraph*{Consent graphs.}

A user might consent for the data they generate on some online service to be shared in the clear with only a specific set of other services. As each user might have their own consent choices for every pair of services, this induces a bipartite trust graph where on one side each vertex corresponds to a (user, service) pair, and on the other side, each vertex corresponds to a service. The presence of an edge between $(u, A)$ and $B$ then indicates that user $u$ consented for their data on service $A$ to be shared in the clear with service $B$. It might also be acceptable for data to be exchanged between a pair of services even if no consent was provied by the user as long as the exchanged data is anonymous. If DP is deemed to meet the bar for data anonymization, then DP on the bipartite trust graph described above would be a better-suited model than central DP (which would not be acceptable privacy-wise due to the lack of a single service that is allowed to see the data in the clear) or local DP (which would result in an unnecessarily poor utility).

\paragraph*{Routing games.} Differentially private routing games \citep{dong2015differential} recommend routes while requiring the origins and destinations of drivers to be considered private. In this more relaxed setting, each driver may be willing to share their origin and destination with a subset of other drivers. This willingness is demonstrated by the location sharing feature on Google maps, where users share their locations with subsets of other users. 
    
\paragraph*{Disease contact tracing.} The COVID-19 pandemic spurred significant development of contact tracing apps and algorithms, along with discussion of privacy concerns \citep{cho2020contact,park2020information}. Trust graph notions may be applicable in contact tracing settings. For example, consider a contact tracing app that alerts users if they have been in recent contact with someone who has contracted the disease. Alice may willing to share her disease status directly with a select close contact circle (such as household members), but still require that
strangers cannot discover Alice's disease status given what the mechanism outputs for them. 
    
\paragraph*{Personalized shopping recommendations.} Consider an advertisement recommender system that provides Alice with recommendations based on her browsing history. It is important that strangers are not able to deduce Alice's browsing history through the  recommendations they are provided. This is particularly important if Alice's browsing history is sensitive, e.g., contains medical information. However, Alice may be willing to share parts of her personal shopping history with trusted close contacts (e.g., family or a partner), and the household may perhaps even benefit from some controlled degree of data sharing in recommendations.

\section{Missing Proofs}
\label{app:proofs}

\begin{proof}[Proof of Lemma~\ref{lem:snb-dp}]
We may write $Z = Z^1 + Z^2$ where $Z^1 \sim \sNB\left(1, 1 - e^{-\eps/\Delta}) = \DLap(\Delta/\eps\right)$ and $Z^2 \sim \sNB\left(r - 1, 1 - e^{-\eps/\Delta}\right)$ are independent. We again can think of $x + Z$ and $x' + Z$ as a post-processing of $x + Z^1$ and $x' + Z^1$, respectively. This implies that
\begin{align*}
&D_{\infty}\left(x + Z ~\|~ x + Z\right) \\
&\leq D_{\infty}\left(x + Z_1 ~\|~ x' + Z_1\right) \\
&= \max_{o \in \bN} \ln\left(\frac{\Pr_{Z_1 \sim \DLap(\Delta/\eps)}[x + Z_1 = o]}{\Pr_{Z_1 \sim \DLap(\Delta/\eps)}[x' + Z_1 = o]}\right) \\
&= \frac{\eps}{\Delta} \cdot \left(|o - x| - |o - x'|\right) \\
& \leq \frac{\eps}{\Delta} \cdot |x' - x| 
\leq \eps,
\end{align*}
where the last two inequalities follow from the triangle inequality and since $0 \leq x, x' \leq \Delta$ respectively.
\end{proof}

\begin{proof}[Proof of~\Cref{thm:lp-protocol-robust}]
Let $\by = (y_u)_{u \in V}$ denote any solution to LP in \eqref{eq:lp-robust}.  We use exactly the same protocol as in \Cref{thm:lp-protocol}. The utility analysis is also similar to before. Thus, we only provide the privacy analysis below.

\paragraph*{Privacy Analysis.} Consider any $v \in V, T \in \binom{N(v)}{\leq t_v}$ and $\bx_{-v} \in \{0, \dots, \Delta\}^{V \setminus \{v\}}$. We write $\ba$ as a shorthand for $(a_u)_{u \in V}$.
Notice that $\VIEW^{V \setminus (N[v] \setminus T)}_P(x)$ is exactly $(\bx_{-(N[v] \setminus T)}, \ba(x), (s^u_{v'})_{u \in V \setminus (N[v] \setminus T), v' \in V})$.
We claim that this is a post-processing of $(z_u + s^u_v)_{u \in (N[v] \setminus T)} \cup (s^u_v)_{u \in T}$. This is simply because $\bx_{-(N[v] \setminus T)}$, $(a_u)_{u \in V \setminus N[v]}$, $(s^u_{v'})_{u \in V \setminus N[v], v' \in V}, (s^u_{v'})_{u \in T, v' \in V \setminus \{v\}}$ do not depend on $x_v = x$ at all and are independent of $(z_u + s^u_v)_{u \in (N[v] \setminus T)} \cup (s^u_v)_{u \in T}$; finally, note that $(a_u(x))_{u \in N[v]}$ is a post-processing of $(z_u + s^u_v)_{u \in N[v]}$ since $a_{u}(x) = \left(z_u + s^u_v\right) + \sum_{v' \in N[V] \setminus \{v\} \atop u_{v'} = u} s^u_{v'}$.

Now, since $(s^u_v(x))_{u \in N[v]}$ are random elements of $\Z_q$ that sums to $x$, we also have that $(z_u + s^u_v)_{u \in (N[v] \setminus T)} \cup (s^u_v)_{u \in T}$ are random elements of $\Z_q$ that sums to $x + \sum_{u \in N[v]} z_u$. In other words, $(z_u + s^u_v)_{u \in (N[v] \setminus T)} \cup (s^u_v)_{u \in T}$ is a post-processing of $x + \sum_{u \in (N[v] \setminus T)} z_u$.
The remainder of the privacy proof then proceed similarly to that in \Cref{thm:lp-protocol}, where we now use the fact that $\sum_{u \in (N[v] \setminus T)} z_u$ is distributed as $\sNB\left(\sum_{u \in (N[v] \setminus T)} y_u, 1 - e^{-\eps/\Delta}\right)$ and $\sum_{u \in (N[v] \setminus T)} y_u \geq 1$ from \eqref{eq:lp-robust}.
\end{proof}

\Cref{thm:robust-lb-generic} is an immediate consequence of the following reduction to the LDP model, similar to \Cref{thm:tgdplb} and the corresponding reduction (\Cref{lem:red-lb}).
\begin{lemma}
\label{lem:robust-lb}
Suppose that there is an $(\eps, G, \bt)$-RTGDP protocol for integer aggregation. Then, there exists an $\eps$-local DP protocol for integer aggregation with the same MSE for $\optind^{\bt}$ users, where $\optind^{\bt}$ denotes the size of the largest $(2, \bt)$-robust packing  $G$.
\end{lemma}

\begin{proof}[Proof of~\Cref{lem:robust-lb}]
Let $U = \{u_1, \dots, u_m\} \subseteq V$ and $(T_v)_{v \in U}$  be the largest $(2, \bt)$-robust packing in $G$ where $m = \optind$. To avoid ambiguity, we use $\tbx = (\tx_1, \dots, \tx_m)$ to denote the input to the local DP protocol.

To construct the LDP protocol, let $Q_1, \dots, Q_m \subseteq V$ be any partition of $V$ such that $(N[u_i] \setminus T_{u_i}) \subseteq Q_i$ for all $i \in [m]$. Such a partition exists because $N[u_1] \setminus T_{u_1}, \dots, N[u_m] \setminus T_{u_m}$ are disjoint by definition of a packing. Let $P$ be any $(\eps, G, \bt)$-RTGDP protocol for integer aggregation. Our LDP protocol $\tP$ works simply by running the protocol $P$ where each user $i \in [m]$ assumes the role of all users in $Q_i$ where the input is defined as
\begin{align*}
x_u =
\begin{cases}
x_i &\text{ if } u = u_i \\
0 &\text{ otherwise.}
\end{cases}
\end{align*}
for all $u \in Q_i$. We then output the estimate as produced by $P$. The MSE of $\tP$ is obviously the same as that of $P$.

To see that this satisfies $\eps$-LDP, consider any $i \in [m], \tbx_{-i} \in \cX^{[m] \setminus \{i\}}$, we have $\VIEW^{[m] \setminus \{i\}}_{\tP}(\tx) = \VIEW^{V \setminus Q_i}_P(\bx(\tx))$ where $\bx(\tx)$ is the input to $P$ as defined above. Since $V \setminus Q_i \subseteq V \setminus (N[u_i] \setminus T_{u_i})$, $\VIEW^{V \setminus Q_i}_P(\bx(\tx))$ is a post-processing of $\VIEW^{V \setminus (N[u_i] \setminus T_{u_i})}_P(\bx(\tx))$. Thus, for every $\tx_i, \tx'_i \in \cX$,  \Cref{lem:postp} implies that
\begin{align*}
&D_{\infty}\left(\VIEW^{[m] \setminus \{i\}}_{\tP}(\tx_i) ~\|~ \VIEW^{[m] \setminus \{i\}}_{\tP}(\tx'_i)\right) \\
&\leq D_{\infty}\left(\VIEW^{V \setminus (N[u_i] \setminus T_{u_i})}_P(\tx_i) ~\|~ \VIEW^{V \setminus (N[u_i] \setminus T_{u_i})}_P(\tx'_i)\right) \\
&\leq \eps,
\end{align*}
where the last inequality is due to $P$ being $(\eps, G, \bt)$-RTGDP protocol.

As a result, $\tP$ is $\eps$-LDP as desired.
\end{proof}

\begin{proof}[Proof of \Cref{thm:robust-bicriteria-gap}]
We will write $\bt'$ as a shorthand for $\bt + \lceil \alpha \cdot \bdeg_U\rceil$.

We claim that $\optp^{\bt'} \geq \gamma \cdot \optlp^{\bt}$ for $\gamma = \frac{\alpha}{8}$. To prove this, first observe that LP duality implies that $\optlp^{\bt}$ is also equal to the optimal of the following LP:
\begin{align} \label{eq:dual-lp-robust}
&\max \sum_{v \in V, T \in \binom{N(v)}{\leq t_v}} w_{v, T} \\
&\text{s.t.} \sum_{u \in V, T \in \binom{N(u)}{\leq t_u} \atop (N[u] \setminus T) \ni v} w_{u, T} \leq 1 &\forall v \in V\label{eq:dual-sum-const} \\
& \qquad 0 \leq w_{v, T} \leq 1 &\forall v \in V, T \in \binom{N(v)}{\leq t_v}. \nonumber
\end{align}

For $T' \subseteq N[v']$, we let $\oT'[v'] := N[v'] \setminus T'$.

Given an optimal solution $(w_{v, T})_{v \in V, T \in \binom{N(v)}{\leq t_v}}$ to the dual LP \eqref{eq:dual-lp-robust}, we construct a $(2, \bt')$-robust packing as follows:
\begin{itemize}
\item Include each $(v, T)$ in the set $R_0$ with probability $2\gamma \cdot w_{v, T}$. 
\item Filter elements in $R_0$ to create $R_1$ where $R_1$ only includes $(v, T)$ such that
\begin{align} \label{eq:vertex-not-included}
v \notin \bigcup_{(v', T') \in (R_0 \setminus \{(v, T)\})} \oT'[v'],
\end{align}
and
\begin{align} \label{eq:subset-correction-small}
\left|\bigcup_{(v', T') \in (R_0 \setminus \{(v, T)\})} (N(v) \cap \oT'[v'])\right| \leq \alpha \cdot \deg(v). 
\end{align}
\item Construct $R_2$ by adding $\left(v, T \cup \bigcup_{(v', T') \in R_1 \setminus \{(v, T)\}} (N(v') \cap \oT'[v'])\right)$, for each $(v, T) \in R_1$.
\end{itemize}
By the conditions imposed when constructing $R_1$, it is simple to see that $R_2$ is a valid $(2, \bt)$-robust packing. Thus, we are only left to argue that it has a large size. To do so, first observe that
\begin{align}
&\E[~|R_2|~] \nonumber \\
&= \E[~|R_1|~] \nonumber \\
&= \sum_{v \in V, T \in \binom{N(v)}{\leq t_v}} \Pr[(v, T) \in R_1] \nonumber \\
&= \sum_{v \in V, T \in \binom{N(v)}{\leq t_v}} \Pr[(v, T) \in R_0] \cdot \Pr[(v, T) \in R_1 \mid (v, T) \in R_0] \nonumber \\
&= \sum_{v \in V, T \in \binom{N(v)}{\leq t_v}} (2\gamma \cdot w_{v, T}) \cdot \Pr[(v, T) \in R_1 \mid (v, T) \in R_0], \label{eq:expectation-expand}
\end{align}
where the last equality is due to the randomized rounding in the first step.

Let us now fix $v \in V$ and $T \in \binom{N(v)}{\leq t_v}$.
To bound $\Pr[(v, T) \in R_1 \mid (v, T) \in R_0]$, note that from our procedure, we have
\begin{align}
&\Pr[(v, T) \in R_1 \mid (v, T) \in R_0] \nonumber \\
&= \Pr\left[\eqref{eq:vertex-not-included} \text{ and } \eqref{eq:subset-correction-small} \text{ both hold}\right] \nonumber \\
&\geq 1 - \Pr[\eqref{eq:vertex-not-included} \text{ fails}] - \Pr[\eqref{eq:subset-correction-small} \text{ fails}]. \label{eq:failure-bound-expand}
\end{align}
We bound each of the two terms separately. For the first term, we have
\begin{align}
&\Pr[\eqref{eq:vertex-not-included} \text{ fails}] \nonumber \\
&= \Pr\left[v \in \bigcup_{(v', T') \in (R^0 \setminus \{(v, T)\})} \oT'[v']\right] \nonumber \\
&= \Pr\left[\exists {v' \in V, T' \in \binom{N(v')}{\leq t_{v'}} \atop (N[v'] \setminus T') \ni v}, (v', T') \in R^0\right] \nonumber \\
&\leq \sum_{v' \in V, T' \in \binom{N(v')}{\leq t_{v'}} \atop (N[v'] \setminus T') \ni v} \Pr[(v', T') \in R^0] \nonumber \\
&= \sum_{v' \in V, T' \in \binom{N(v')}{\leq t_{v'}} \atop (N[v'] \setminus T') \ni v} 2\gamma \cdot w_{v', T'} \nonumber \\
&\overset{\eqref{eq:dual-sum-const}}{\leq} 2\gamma \nonumber \\
&\leq \frac{1}{4}. \label{eq:first-failure-bound}
\end{align}
Similarly, we have
\begin{align}
&\Pr[\eqref{eq:subset-correction-small} \text{ fails}] \nonumber \\
&= \Pr\left[\left|\bigcup_{(v', T') \in (R^0 \setminus \{(v, T)\})} (N(v) \cap \oT'[v'])\right| > \alpha \cdot \deg(v)\right] \nonumber \\
&\leq \frac{\E\left[\left|\bigcup_{(v', T') \in (R^0 \setminus \{(v, T)\})} (N(v) \cap \oT'[v'])\right|\right]}{\alpha \cdot \deg(v)} \nonumber \\
&=  \frac{\sum_{u \in N(v)} \Pr\left[u \in \bigcup_{(v', T') \in (R^0 \setminus \{(v, T)\})} \oT'[v']\right]}{\alpha \cdot \deg(v)}\nonumber \\
&= \frac{\sum_{u \in N(v)} \Pr\left[\exists {v' \in V, T' \in \binom{N(v')}{\leq t_{v'}} \atop (N[v'] \setminus T') \ni u}, (v', T') \in R^0]\right]}{\alpha \cdot \deg(v)} \nonumber \\
&\leq \frac{\sum_{u \in N(v)} \sum_{v' \in V, T' \in \binom{N(v')}{\leq t_{v'}} \atop (N[v'] \setminus T') \ni u} \Pr[(v', T') \in R^0]}{\alpha \cdot \deg(v)} \nonumber \\
&= \frac{\sum_{u \in N(v)} \sum_{v' \in V, T' \in \binom{N(v')}{\leq t_{v'}} \atop (N[v'] \setminus T') \ni u} 2\gamma \cdot w_{v', T'}}{\alpha \cdot \deg(v)} \nonumber \\
&\overset{\eqref{eq:dual-sum-const}}{\leq} \frac{\sum_{u \in N(v)} 2\gamma}{\alpha \cdot \deg(v)} \nonumber \\
&= \frac{2\gamma}{\alpha} \leq \frac{1}{4}, \label{eq:second-failure-bound}
\end{align}
where the first inequality is due to Markov and the last inequality is from our choice of $\gamma$.

Combining \Cref{eq:expectation-expand,eq:failure-bound-expand,eq:first-failure-bound,eq:second-failure-bound}, we get
\begin{align*}
\E[~|R_2|~] &\geq \sum_{v \in V, T \in \binom{N(v)}{\leq t_v}} \gamma \cdot w_{v, T} = \gamma \cdot \optlp^{\bt},
\end{align*}
which concludes the proof.
\end{proof}

\section{Additional Experiment Details}
\label{app:exp}

We next provide additional experiment details. Optimization problems were solved using \texttt{cvxpy} \citep{diamond2016cvxpy,agrawal2018rewriting}.

\subsection{Additional Dataset Details}
\label{app:exp:data}

Table \ref{tab:datasets} lists additional details for each dataset. All datasets are publicly available through SNAP \citep{snapnets}. We give additional dataset descriptions below.

\begin{table}[!ht]
\caption{Integrality gap comparison of $\OPT_{\LP}$ to minimum dominating set size $|T|$ for 9 network datasets.}
\label{tab:datasets}
\begin{center}
\begin{small}
\begin{tabular}{lrrr}
\toprule
Dataset & Number of nodes $n$ & Number of edges $|E|$ & Maximum degree \\
\midrule
EU Emails (Core)  & $1,005$ & $16,706$ & $347$ \\
Bitcoin (Alpha) & $3,783$ & $12,972$ & $507$ \\
Facebook & $4,039$ & $88,234$ & $1,045$ \\
Bitcoin (OTC) & $5,881$ & $18,591$ & $788$ \\
Enron Emails   & $36,692$ & $183,831$ & $1,383$  \\
GitHub & $37,700$ & $289,003$ & $9,458$ \\
Epinions & $75,879$ & $405,740$ & $3,044$  \\
Twitter & $81,306$ & $1,342,310$ & $3,383$  \\
EU Emails (All) & $265,214$ & $365,570$ & $7,636$ \\
\bottomrule
\end{tabular}
\end{small}
\end{center}
\vskip -0.1in
\end{table} 

\textit{EU Emails datasets.} This data comes from an email network at a European research institution collected by \citep{leskovec2007graph}. We include an undirected edge between sender and receiver who have exchanged an email, though the original dataset contains directed edge information. We consider two subgraphs: \textit{(i)} a ``core'' subgraph consisting of email addresses within the research institution (which we refer to as \textit{EU Emails (Core)}), and \textit{(ii)} the full graph of all emails contained in the dataset (which we refer to as \textit{EU Emails (All)}). 

\textit{Enron Emails.} This data comes from an email communication network within Enron \citep{klimt2004enron}. The graph contains an undirected edge between a sender and receiver if at least one email was exchanged between them. The original graph is undirected. 

\textit{Bitcoin datasets.} We include network data from two different Bitcoin trading platforms, \textit{Bitcoin OTC} and \textit{Bitcoin Alpha} \citep{kumar2016edge,kumar2018rev2}. In the original data, each user rater another user with a value between $-10$ and $10$, where a negative rating corresponds to mistrust and a positive rating corresponds to trust. In our analysis, we include an undirected edge between users if and only if one user rates another with a value greater than $0$. We ignore mistrust ratings. In general, further analysis that includes the mistrust ratings would be interesting to conduct.

\textit{Facebook dataset.} This data from Facebook was published by \citep{leskovec2012learning}. Each undirected edge indicates a friend relationship between two users. 

\textit{GitHub dataset.} This data comes from a social network of GitHub developers collected by \citep{rozemberczki2019multiscale}. Edges represent mutual follower relationships between two users.

\textit{Epinions dataset.} This data comes from the consumer review site Epinions.com \citep{richardson2003trust}. We include an undirected edge if one user ``trusts'' another by giving a positive rating to the other user (which indicates trusting their reviews). The original dataset is a directed graph.

\textit{Twitter dataset.} This data encodes follower relationships from Twitter \citep{leskovec2012learning}. We include an undirected edge if one user has follows another. The original dataset contains directed edges for follower relationships.


\subsection{Integrality Gap Between $\OPT_{\LP}$ and Minimum Dominating Set}
The proposed LP-based mechanism for achieving Trust Graph DP presented in Theorem \ref{thm:lp-protocol} will always perform at least as well as the dominating set protocol in terms of error. Here we consider whether in practice, the LP-based mechanism is better than the dominating set protocol. We observed an integrality gap between $\OPT_{\LP}$ and the size of the minimum dominating set $|T|$ on three out of nine datasets. Notably, we only observed the integrality gap on the email communication datasets, and not on the Bitcoin or social network datasets. Theoretically, it is known that the integrality gap between $\OPT_{\LP}$ and $|T|$ can be up to a factor of up to $O(\log(n))$~\cite{WS11}. Further study of the integrality gaps that might arise in other real network settings remains an interesting open question.
Table \ref{tab:integrality_gaps} lists the ratio $\frac{\OPT_{\LP}}{|T|}$ for each dataset.

Whether any given graph exhibits an integrality gap or not, a prevailing advantage of the proposed improved algorithm via linear programming for TGDP over a minimum dominating set protocol lies in computational efficiency, as finding the minimum dominating set is NP-hard.

\begin{table}[!ht]
\caption{Integrality gap comparison of $\OPT_{\LP}$ to minimum dominating set size $|T|$ for 9 network datasets.}
\label{tab:integrality_gaps}
\begin{center}
\begin{small}
\begin{tabular}{lrrrr}
\toprule
Dataset & $n$ & $\OPT_{\LP}$ & $|T|$ & $\frac{\OPT_{\LP}}{|T|}$ \\
\midrule
EU Emails (Core)  & $1,005$ & $111.97$ & $128 $& $0.8748$ \\
Bitcoin (Alpha) & $3,783$ & $686$ & $686$ & $1$  \\
Facebook & $4,039$ & $10$ & $10$ & $1$ \\
Bitcoin (OTC) & $5,881$ & $1,126$ & $1,126$ & $1$  \\
Enron Emails   & $36,692$ & $3,060.66$ & $3,062$ & $0.9996$  \\
GitHub & $37,700$ & $4,538$ & $4,538$ & $1$ \\
Epinions & $75,879$ & $15,734$ & $15,734$ & $1$  \\
Twitter & $81,306$ & $961$ & $961$ & $1$  \\
EU Emails (All) & $265,214$ & $18,074.40$ & $18,181$ & $0.9941$  \\
\bottomrule
\end{tabular}
\end{small}
\end{center}
\vskip -0.1in
\end{table}

\subsection{Other Local DP Mechanisms for Aggregation} \label{app:local-mech-agg}

Note that we focused our comparisons to the LDP version of the Laplace mechanism since the ratio has a simple expression. Nevertheless, we point out that there are other LDP mechanisms for aggregation. For example, one can apply randomized rounding to the input (where we set it to $\Delta$ with probability $x_i / \Delta$ and set it to zero otherwise) before applying the classic randomized response (RR) mechanism~\cite{warner1965randomized}.  The MSE of this method is actually input-dependent due to the randomized rounding, making it harder to compare against our method. To give the benefit to this algorithm, let us assume for the moment that all inputs are either 0 or $\Delta$. In this case, there is no error due to the randomized rounding. A simple calculation shows that the error from  here is $c_\epsilon \cdot 2\Delta^2 n / \epsilon^2$ where $c_\epsilon = \frac{e^\epsilon \cdot \epsilon^2}{2(e^\epsilon - 1)^2}$. For $\epsilon \leq 1$, this constant $c_\epsilon$ is at least 0.46. Thus, in all cases, our mechanism still demonstrates a significant improvement over this over-optimistic estimate of the error for randomized rounding + RR.
\section{Relating Packing Number and Minimum Dominating Set}
\label{app:packing-v-domset}

In this section, we prove \Cref{thm:domset-v-packing}.
Our proof follows that of Halldorsson et al.~\cite{HalldorssonKT00}, who gave a greedy algorithm that provides a $\sqrt{n}$-approximation for the packing number. At a high level, the main difference between our proof and theirs is that we compare the greedy solution with the (optimal) LP solution, whereas they compare it with the (optimal) integral solution. The LP for the packing number turns out to be exactly the dual of the LP for the minimum dominating set (i.e., LP \eqref{eq:origlp}); this then gives us the desired claim. 

\begin{proof}[Proof of \Cref{thm:domset-v-packing}]
Consider the dual of LP \eqref{eq:origlp}, which can be written as follows:
\begin{align} \label{eq:lp-pack}
&\min \sum_{u \in V} y_u \\
&\text{s.t.} \sum_{u \in N[v]} y_u \leq 1 &\forall v \in V \label{eq:packing-constraint} \\
& \qquad 0 \leq y_u \leq 1 &\forall u \in V. \nonumber
\end{align}

Let $\by^{*} = (y^{*}_u)_{u \in V}$ denote an optimal solution of the LP \eqref{eq:lp-pack}.
Consider the following greedy algorithm by Halldorsson et al.~\cite{HalldorssonKT00}.
\begin{itemize}
\item Set $S \gets \emptyset, i \gets 0$, and $V_i \gets V$.
\item While $V_i \neq \emptyset$ do:
\begin{itemize}
\item Let $v_i$ be the vertex in $V_i$ with the smallest degree (ties broken arbitrarily).
\item Set $S \gets S \cup \{v_i\}$.
\item Let $Z_i := \{u \in V_i \mid N[u] \cap N[v_i] \ne \emptyset\}$.
\item Set $V_{i + 1} \gets V_i \setminus Z_i$.
\item Set $i \gets i + 1$.
\end{itemize}
\item Output $S$.
\end{itemize}

It is clear that the output $S$ is indeed a packing; let $q = |S|$.  We claim the following for all $i \in [q]$:
\begin{align} \label{eq:lp-v-integral-packing}
\sum_{u \in Z_i} y^{*}_u \leq \sqrt{n}.
\end{align}
Before we prove \eqref{eq:lp-v-integral-packing}, let us argue that it implies $\packnum(G) \geq \OPT_{\LP} / \sqrt{n}$. Notice that $\{ Z_i \}_{i \in [q]}$ is a partition of $V$. Therefore, we have
\begin{align*}
\OPT_{\LP} = \sum_{u \in V} y^{*}_u = \sum_{i \in [q]} \sum_{u \in Z_i} y^{*}_u \overset{\eqref{eq:lp-v-integral-packing}}{\leq} \sum_{i \in [q]} \sqrt{n} \leq \packnum(G) \cdot \sqrt{n},
\end{align*}
as desired.

Finally, we prove \eqref{eq:lp-v-integral-packing}. Consider two cases based on whether $\deg(v_i) \leq \sqrt{n} - 1$.
\begin{itemize}
\item Case I: $\deg(v_i) \leq \sqrt{n} - 1$. In this case, we have
\begin{align*}
\sum_{u \in Z_i} y^{*}_u \leq \sum_{w \in N[v_i]} \sum_{u \in N[w]} y^{*}_u \overset{\eqref{eq:packing-constraint}}{\leq} \sum_{w \in N[v_i]} 1 = \deg(v_i) + 1 \leq \sqrt{n}.
\end{align*}
\item Case II: $\deg(v_i) > \sqrt{n} - 1$. Since we pick $v_i$ to be the vertex with the smallest degree among those in $V_i$, we have $\deg(u) > \sqrt{n} - 1$ for all $u \in Z_i$. Therefore, we have
\begin{align*}
\sum_{u \in Z_i} y^{*}_u &\leq \frac{1}{\sqrt{n}} \cdot \sum_{u \in Z_i} (\deg(u) + 1) \cdot y^{*}_u \\
&\leq \frac{1}{\sqrt{n}} \cdot \sum_{u \in V} (\deg(u) + 1) \cdot y^{*}_u \\
&= \frac{1}{\sqrt{n}} \cdot \sum_{v \in V} \sum_{u \in N[v]} y^{*}_u \\
&\overset{\eqref{eq:packing-constraint}}{\leq} \frac{1}{\sqrt{n}} \cdot n \\
&= \sqrt{n}.
\end{align*}
\end{itemize}
Thus, in both cases \eqref{eq:lp-v-integral-packing} holds, which completes our proof.
\end{proof}

\end{document}